\documentclass[runningheads]{llncs}
\usepackage[T1]{fontenc}

\usepackage{hyperref}

\usepackage{xspace}
\usepackage{amsmath}
\usepackage{amsfonts}
\usepackage{mathtools}
\usepackage{amssymb}
\usepackage{cleveref}
\usepackage[inline]{enumitem}

\usepackage{tabularx}
\usepackage{booktabs}
\usepackage{multirow}
\usepackage{a4wide}

\usepackage{tikz}
\usetikzlibrary{calc}

\newcommand{\ARH}{{\normalfont\textsc{Anonymous Refugee Housing with Upper-Bounds}}\xspace}
\newcommand{\ARHshort}{{\normalfont\textsc{ARH-UB}}\xspace}
\newcommand{\RARH}{{\normalfont\textsc{Relaxed \ARHshort}}\xspace}
\newcommand{\RARHshort}{{\normalfont\textsc{RARH-UB}}\xspace}

\newcommand{\R}{\ensuremath{R}\xspace}
\newcommand{\I}{\ensuremath{I}\xspace}
\newcommand{\ub}{\ensuremath{\operatorname{ub}}}
\newcommand{\assgn}{\ensuremath{\iota}}
\newcommand{\plc}{\ensuremath{\pi}}

\newcommand{\N}{\ensuremath{\mathbb{N}}}

\newcommand{\tw}{\ensuremath{\operatorname{tw}}}
\newcommand{\fes}{\ensuremath{\operatorname{fes}}}

\newcommand{\Oh}[1]{\ensuremath{{\mathcal{O}\left(#1\right)}}}

\newcommand{\NP}{\textsf{NP}\xspace}
\newcommand{\NPh}{\NP-hard\xspace}
\newcommand{\NPc}{\NP-complete\xspace}
\newcommand{\FPT}{\textsf{FPT}\xspace}
\newcommand{\XP}{\textsf{XP}\xspace}
\newcommand{\W}[1][1]{\textsf{W[#1]}\xspace}
\newcommand{\Wh}[1][1]{\W[#1]-hard\xspace}

\newcommand{\Whness}[1][1]{\W[#1]-hardness\xspace}
\newcommand{\pNPh}{para-\NP-hard\xspace}

\newcommand{\No}{\emph{no}\xspace}
\newcommand{\Yes}{\emph{yes}\xspace}

\newcommand{\YesI}{\Yes-instance\xspace}

\AtEndEnvironment{proof}{\phantom{}\qed} %
\newtheorem{claim}{Claim}
\crefname{claim}{Claim}{Claims}
\newenvironment{claimproof}[1][Proof]{\paragraph{#1.}}{\hfill$\blacktriangleleft$\\\medskip}
\begin{document}
\title{Towards More Realistic Models for Refugee Integration: Anonymous Refugee Housing with Upper-Bounds}
\titlerunning{Anonymous Refugee Housing with Upper-Bounds}
\author{Šimon Schierreich\orcidID{0000-0001-8901-1942}}%
\authorrunning{Šimon Schierreich}
\institute{Czech Technical University in Prague, Prague, Czechia\\
\email{schiesim@fit.cvut.cz}}
\maketitle              %
\begin{abstract}
	[Knop and Schierreich; AAMAS~'23] recently introduced a novel model for refugee housing, where we are given a city's topology represented as a graph, a set of inhabitants who occupy some vertices of the topology, and a set of refugees that we need to house in unoccupied vertices. Moreover, each inhabitant approves specific numbers of refugees in his neighbourhood, and our goal is to find housing such that every inhabitant approves the number of refugees housed in his neighbourhood. However, such a model suffers from several problems: (i) it is computationally hard to find desirable housing, and (ii) the inhabitants' preferences are not necessarily continuous, which hardly occurs in practice. To avoid these objections, we introduce a restricted variant of the problem that is more closely aligned with practical scenarios of refugee housing and study its complexity.

\end{abstract}

\section{Introduction}

According to the numbers of the United Nations High Commissioner for Refugees (UNHCR), in May 2022, we passed the tragic milestone of 100 million forcibly displaced people around the world~\cite{UNHCR2022}. As an absolute number may not be so telling, an easy calculation gives us that 1 in every 78 people on earth has been forced to leave their home. Among the most common causes are armed and political conflicts, food insecurity, and natural disasters that go hand in hand with climate change~\cite{FischerSK2021,UNHCR2022}. Furthermore, UNHCR reports~\cite{UNHCR2022} that 83\% of all forcibly displaced persons are hosted in low- and middle-income countries. Furthermore, around 69\% of the refugees originate from only five countries and, at the same time, 72\% of them are hosted in countries neighbouring the country of origin. This clearly indicates that the burden of refugees is distributed very unevenly and unfairly~\cite{Schuck1997,NRC2022}.

Problems with refugee redistribution and housing have received the attention of computer scientists and computational social choice community in particular even very recently. The first results in this direction were given in 2016 in the technical report by Delacrétaz et al. (the report was recently published in \cite{DelacretazKT2023}) who proposed a novel model of \emph{refugee redistribution} based on double-sided matching markets. In this model, we are given a set of locations together with multidimensional constraints such as the number of refugees that can be accepted. In addition, we have a set of refugees with multidimensional features, such as family size. This model is very suitable for \emph{global phase of refugee resettlement}, when incoming refugees need to be distributed between different communities to decrease the burden of the most affected ones. This line of research was later followed by several papers~\cite{AhaniGPTT2021,AndersonE2020}. The most important works for us are the papers of Aziz et al.~\cite{AzizCGS2018} and Kuckuck et al.~\cite{KuckuckRW2019} that investigate stability, Pareto optimality, and associated computational issues. Bansak et al.~\cite{BansakFHDHLW2018} and Ahani et al.~\cite{AhaniAMTT2021} also investigated ML techniques that can significantly improve the performance of existing redistribution mechanisms.

Once refugees are redistributed using, for example, the refugee redistribution mechanism of \cite{DelacretazKT2023}, it remains to assign incoming refugees to empty houses of the host community. Real-world examples~\cite{AgerS2008,OECD2018,ZierschDW2023} clearly show that housing incoming refugees so that residents do not feel threatened by newcomers and so that refugees are ``happy'' about their new homes is one of the keys to successful integration. For this purpose, Knop and Schierreich~\cite{KnopS2023} recently introduced a computational model for successful \emph{local phase of refugee resettlement}. The simplest setting they propose is called \emph{anonymous refugee housing} and is defined as follows: The city is represented as an undirected graph (called \emph{topology}) with the vertices of this graph being houses and the edges representing the neighbour relationship between houses. Then some houses are occupied by inhabitants who approve some numbers of refugees in their neighbourhood. Finally, there is a set of refugees, and the goal is to assign refugees to empty houses of the topology so that each inhabitant approves the number of refugees in his or her neighbourhood. In this initial work, the authors focus mainly on computational issues related to the existence of such \emph{inhabitants-respecting housings}, which can be seen as a refinement of the well-known \emph{individual rationality} for refugee housing.

However, the model of Knop and Schierreich~\cite{KnopS2023} described above is not without limitations, which raise valid concerns from several angles. First, in the setting with anonymous preferences, refugees' opinions about housing are not taken into account at all. This omission can have undesirable consequences, including secondary migration~\cite{AgerS2008,ZierschDW2023}, which, in turn, can hinder the successful integration of newcomers into host communities. To dispel this criticism, Knop and Schierreich~\cite{KnopS2023} also introduced two more advanced models, where even refugees have preferences. 
A more fundamental weakness lies in the assumption of approval preferences, which may not accurately reflect the realities of refugee housing. Specifically, in their model the preferences of inhabitants can be non-continuous; for example, an inhabitant may accept either $2$ or $4$ refugees but not $3$ or even no refugees to his neighbourhood. Moreover, this property was significantly exploited in the intractability results of \cite{KnopS2023}. Consequently, it is natural to question whether there exist more realistic preference models and whether these models can lead to more positive algorithmic outcomes compared to the original setting.

\paragraph{Our Contribution.} 

In this work, we introduce a variant of anonymous refugee housing with restricted preferences. Specifically, we require preferences to exhibit two key characteristics: to be continuous and to always contain zero. That is, in our model, each inhabitant declares only the upper-bound on the number of refugees in his or her neighbourhood. Such preferences bring the model of refugee housing closer to practice, and even the collection of such preferences is much better realisable in practice. It should be noted that the investigation of various preferences restrictions is widespread in computational social choice~\cite{ElkindL2015,DoucetteC2017,TerzopoulouKO2021,ChengR2022,ElkindLP2022}. In fact, our upper-bounds are a special case of candidate-interval preferences~\cite{ElkindL2015}, which are the approval counterpart of classical single-peaked preferences known from ordinal voting~\cite{Black1958} and have been extensively studied in recent years~\cite{BrillIMP2022,ElkindLZ2022,KusekBFKK2023}.

We start our study with a formal definition of the computational problem of interest. Then, we derive various properties that are of independent interest or useful in future sections. Next, we focus on the equilibrium guarantees and prove necessary conditions for its existence and show that equilibrium is not guaranteed to exist even in very simple instances.%
Motivated by the equilibrium non-existence, we then further investigate the associated computational problem that asks for the complexity of deciding an equilibrium existence. We start with a perspective of classical computational complexity and show that deciding the existence of equilibrium is computationally hard, even if the topology is a planar graph of constant degree. Next, we focus on the instances where the topology is restricted. To this end, we provide polynomial-time algorithms for trees, graphs of maximum degree at most~$2$, and complete bipartite graphs. Then, we study additional restrictions of the topology that are, in some sense, not far from the trivial polynomial-time solvable cases. In particular, we use the framework of parameterised complexity\footnote{Formal introduction to parameterised complexity is provided in \Cref{sec:preliminaries}.} to show that the problem is fixed-parameter tractable when parameterised by the feedback-edge set number, by the number of edges that need to be added to the graph to become a complete bipartite graph, or by the size of the modulator to complete bipartite graph. Additionally, we show that the problem is in \XP and is \Wh with respect to the tree-width of the topology. %
Finally, we propose a relaxation of the inhabitants-respecting housing, which shows up to be a very strong requirement that fails to exist even in very simple instances. Specifically, we are interested in finding a housing that cumulatively exceeds the agents' upper-bounds by at most $t$. This relaxation can also be seen as an approximation of the original anonymous refugee housing with upper-bounds. We again study guarantees, structural properties, and computational complexity of the associated computational problem, to which end we give several algorithmic upper-bounds and complexity lower-bounds.

\paragraph{Related Work.} The model of refugee housing is related to the famous Schelling model of racial segregation~\cite{Schelling1969,Schelling1971} and, more precisely, to its graph-theoretic refinements~\cite{ChauhanLM2018,AgarwalEGISV2021}. However, in the model of Agarwal et al.~\cite{AgarwalEGISV2021}, the preferences of the agents are based on types of agents and are implicit: the fewer agents of other types in the neighbourhood, the better. In addition, agents can strategically move around the topology to improve their position. Schelling games have received a lot of attention from the perspective of computation and complexity~\cite{EchzellFLMPSSS2019,KreiselBFN2022,BiloBLM2022b,BiloBLM2022a,BiloBDLM02023,DeligkasEG2023}.

Next, in \emph{hedonic seat arrangement}~\cite{BodlaenderHJOOZ2020,Wilczynski2023,Ceylan0R2023}, the goal is also to find an assignment of agents to the vertices of some underlying topology. However, the preferences of the agents are much more verbose; they depend on the identity of each agent. Moreover, in our setting, the inhabitants already occupy some vertices of the topology.
The assignment of agents to vertices of a topology is a main task also in multiple network games. In recently introduced \emph{topological distance games}~\cite{BullingerS2023,DeligkasEKS2024}, agents' utilities in some allocation are based on their inherent valuation of other agents and on their distances with respect to the allocation.

Related are also different classes of coalition formation games. Arguably, the most influential model in this direction are \emph{hedonic games}~\cite{DrezeG1980}, where the goal is to partition a set of agents into coalitions subject to agents' preferences that are based on the identity of other agents assigned to the same coalition. In our setting, the satisfaction of agents does not depend on the identity of agents, but on the coalition size; as in \emph{anonymous hedonic games}~\cite{Banerjee2001}. However, in both settings of hedonic games, the coalitions are pairwise disjoint, while in refugee housing the neighbourhoods are overlapping.

\section{Preliminaries}\label{sec:preliminaries}

We use~$\N$ to denote the set of positive integers. For two distinct positive integers~$i,j\in\N$, where~$i < j$, we define~$[i,j] = \{i,i+1,\ldots,j-1,j\}$ and call it an \emph{interval}, and we set~$[i] = [1,i]$ and~$[i]_0 = [0,i]$. For a set~$S$, by~$2^S$ we denote the set of all subsets of~$S$ and, given~$k\leq |S|$, we denote by~$\binom{S}{k}$ the set of all~$k$-sized subsets of~$S$.

\subsection{Graph Theory}
We follow the monograph of Diestel~\cite{Diestel2017} for basic graph-theoretical terminology. A~simple undirected graph is a pair~$G=(V,E)$, where~$V$ is a set of \emph{vertices} and~$E\subseteq \binom{V}{2}$ is a set of \emph{edges}. Given a vertex~$v\in V$, we define its \emph{(open) neighbourhood} as a set~$N_G(v) = \{ u \in V \mid \{u,v\}\in E(G)\}$ and the \emph{closed neighbourhood} as a set~$N_G[v] = N_G(v) \cup \{v\}$.
The \emph{degree} of a vertex~$v\in V$ is~$\deg(v) = |N_G(v)|$.
We extend the notion of neighbourhoods to sets as follows. Let~$S\subseteq V$ be a set. The neighbourhood of~$S$ is~$N_G(S) = \bigcup_{v\in S} N_G(v)$.%
Given a graph $G$ and a set $S\subseteq V(G)$, by $G\setminus S$ we denote the graph $(V(G)\setminus S,E(G)\cap\binom{V\setminus S}{2})$.

\subsection{Refugee Housing}
By~$\I$ we denote a set of \emph{inhabitants} and by~$\R\in\N$ we denote a number of \emph{refugees}. A \emph{topology} is a simple undirected graph~$G=(V,E)$ such that~$|V| \geq |\I| + \R$. An \emph{inhabitants assignment} is an injective function~$\assgn\colon \I\to V$ and by~$V_\I$ we denote the set of all vertices occupied by the inhabitants, that is,~$V_\I = \{v\in V\mid \exists i\in\I \colon \assgn(i) = v\}$. The set of \emph{unoccupied} or \emph{empty} vertices is~$V_U = V \setminus V_\I$.
The preferences of the inhabitants are expressed in the maximum number of refugees they allow to be housed in their neighbourhoods. Formally, we are given an \emph{upper-bound}~$\ub(i)\in [\deg(\assgn(i))]_0$ for every inhabitant~$i\in\I$.
A \emph{housing} is a subset~$\pi\subseteq V_U$ and we say that $\pi$ is \emph{inhabitants-respecting}, if we have~$N_G(\assgn(i))\cap \pi \leq \ub(i)$ for every~${i\in\I}$.

\subsection{Parameterised Complexity}
In parameterised algorithmics~\cite{Niedermeier2006,DowneyF2013,CyganFKLMPPS2015}, we investigate the computational complexity of a computational problem not only with respect to the input size~$n$, but we also take into account some \emph{parameter}~$k$ that additionally restricts the input instance. The goal is then to find a parameter (or a set of parameters) that is small in real-world instances and such that the ``hardness'' can be confined to it.
The most favourable outcome is an algorithm with running time~$f(k)\cdot n^c$, where~$f$ is any computable function and~$c$ is a constant independent of the parameter. We call such algorithms \emph{fixed-parameter} algorithms, and the complexity class containing all problems admitting algorithms with this running time is called \FPT. The less favourable outcome, but still positive, is a complexity class \XP that contains all the problems for which there exists an algorithm with running time~$n^{f(k)}$, where~$f$ is, again, a computable function.
Not all parameterised problems admit \FPT or even \XP algorithms. To exclude the existence of a fixed-parameter algorithm, we can show that our problem is \Wh[1]. This can be done via the so-called \emph{parameterised reduction} from any \Wh[1] problem. If a parameterised problem is \NPh even for a constant value of a parameter~$k$, we say that the problem is \pNPh and, naturally, such problems cannot admit an \XP algorithm, unless~$\textsf{P}=\NP$.
For a more comprehensive introduction to the framework of parameterised complexity, we refer the reader to the monograph of Cygan et al.~\cite{CyganFKLMPPS2015}.

In this work, we also provide conditional lower-bounds on the running time of our algorithms. The first hypothesis we use is the well-known Exponential Time Hypothesis (ETH) of Impagliazzo and Paturi~\cite{ImpagliazzoP2001} which, roughly speaking, states that there is no algorithm solving \textsc{$3$-SAT} in time sub-exponential in the number of variables. %
We refer the interested reader to the extraordinary survey of Lokshtanov et al.~\cite{LokshtanovMS2018} for a more thorough introduction to ETH.

\subsection{Structural Parameters}
Let~$G=(V,E)$ be a graph. A set~$M\subseteq V$ is called a \emph{vertex cover} of~$G$ if~$G\setminus M$ is an edgeless graph. The \emph{vertex cover number}, denoted~$\operatorname{vc}(G)$, is the minimum size vertex cover in~$G$.

For a graph $G=(V,E)$ and $F\subseteq E$, the set $F$ is \emph{feedback-edge set} if the graph $(V,E\setminus F)$ is an acyclic graph. The \emph{feedback-edge set number} is the minimum size feedback-edge set in $G$.

A \emph{tree decomposition} of a graph~$G=(V,E)$ is a triple~$\mathcal{T}=(T,\beta,r)$, where~$T$ is a tree rooted in~$r$, and~$\beta\colon V(T)\to 2^V$ is a mapping satisfying
\begin{enumerate}
	\item~$\bigcup_{x\in V(T)} \beta(x) = V$, that is, every vertex~$v\in V$ is mapped to at least one~$x\in V(T)$,
	\item for every edge~$\{u,v\}\in E$ there exists a node~$x\in V(T)$ such that~$u,v\in\beta(x)$, and
	\item for every~$v\in V$ the nodes~$\{x\in V(T)\mid v\in\beta(x)\}$ form a connected sub-tree of~$T$.
\end{enumerate}
To distinguish between the vertices of~$G$ and the vertices of~$T$, we use the term \emph{node} when referring to the vertices of tree decomposition. Additionally, for a node~$x\in V(T)$ we refer to the set~$\beta(x)$ as a \emph{bag}.

The \emph{width} of a tree decomposition~$\mathcal{T}$ is~$\max_{x\in V(T)} |\beta(x)| - 1$ and the \emph{tree-width} of a graph~$G$, denoted~$\operatorname{tw}(G)$, is the minimum width of a tree decomposition of~$G$ over all possible tree decompositions.

\section{The Model}\label{sec:model}

In this section, we formally define the problem of our interest and provide some basic properties and observations regarding the existence of inhabitants-respecting housing. We start with the introduction of the computational problem of our interest -- the \ARH problem.

\vspace{0.15cm}
\noindent\begin{tabularx}{\linewidth}{lX}
	\toprule
	\multicolumn{2}{c}{\small\ARH} \\\midrule
	\small\emph{Input:} & \small{} A topology $G$, a set of inhabitants $\I$, a number of refugees $\R$, an inhabitants assignment~$\assgn$, and an upper-bound function $\ub$. \\
	\small\emph{Question:} & \small{}Is there a set $\pi\subseteq V_U$ of size $R$ such that $\pi$ is an inhabitants-respecting housing? \\
	\bottomrule
\end{tabularx}
\vspace{0.15cm}

We use \ARHshort to abbreviate the problem name. To illustrate the definition of the problem, in what follows, we give an example of an instance and a few possible solutions.

\begin{example}\label{ex:model}\itshape
	Let $\mathcal{I}$ be as shown in \Cref{fig:example}, that is, we have three inhabitants $h_1$, $h_2$, and $h_3$ with upper-bounds $\ub(h_1) = 1$, $\ub(h_2) = 2$, and $\ub(h_3) = 1$, and two refugees to house. The empty vertices are $x$, $y$, and $z$.
	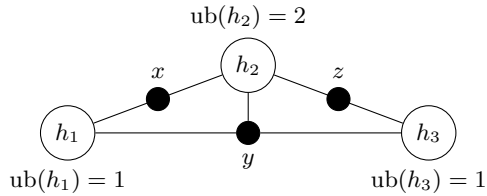
\begin{figure}[tb!]
		\centering
		\begin{tikzpicture}[scale=0.6]
			\node[draw,circle,label={270:\small$\ub(h_1) = 1$}] (h1) at (0,0) {$h_1$};
			\node[draw,circle,label={90:\small$\ub(h_2) = 2$}] (h2) at (4,1.5) {$h_2$};
			\node[draw,circle,label={270:\small$\ub(h_3) = 1$}] (h3) at (8,0) {$h_3$};
			\node[draw,circle,label={270:$y$},fill=black] (u1) at (4,0) {};
			
			\draw (h1) -- (u1) -- (h3);
			\draw (h2) -- (u1);
			\draw (h1) -- (h2) node[midway,draw,circle,fill=black,label={90:$x$}] {};
			\draw (h2) -- (h3) node[midway,draw,circle,fill=black,label={90:$z$}] {};
		\end{tikzpicture}
		\caption{An example of the \ARHshort problem. The filled vertices are empty and can be used to house refugees, while the named vertices are occupied by inhabitants. Labels next to inhabitant-occupied vertices represent upper-bounds of the corresponding inhabitant.}
		\label{fig:example}
	\end{figure}
	Observe that if we house a refugee on vertex $y$, than all inhabitants are satisfied as their upper-bound is at least~$1$. However, we cannot place the second refugee in any other empty vertex as it will exceed either $h_1$'s or $h_3$'s upper-bound. Therefore, we cannot house a refugee in $y$.
	On the other hand, if we choose as a solution the vertices $x$ and $z$, this is clearly a solution. The inhabitant $h_2$ is the only one with more than one refugee in her neighbourhood and it matches her upper-bound.
\end{example}

Observe that we can freely assume that $G$ is a bipartite graph with one part consisting of vertices occupied by inhabitants and the second part consisting of empty houses. In this setting, refugees do not have preferences at all, and inhabitants are interested only in the number of refugees in their neighbourhood. Therefore, we can remove edges between two vertices of the same type without changing the solution (cf. \cite[Prop. 3.5]{KnopS2023}).

Next, we show that, despite the simplicity of the model, the existence of an inhabitants-respecting housing is not guaranteed.

\begin{proposition}\label{lem:housing-non-existence-intolerant}
	There is an instance of the \ARHshort problem with no inhabitants-respecting refugees housing.
\end{proposition}
\begin{proof}
	Let the topology~$G$ be a star with~$n$ leaves and~$h$ be the single inhabitant assigned to the centre of~$G$. Furthermore, let $\ub(h) = 0$ and $\R \geq 1$. We can house our refugees only on the leaves of~$G$, but all are neighbours to the vertex $\assgn(h)$. Therefore, there is no inhabitants-respecting housing.
\end{proof}

One can argue that in \Cref{lem:housing-non-existence-intolerant}, we use inhabitant~$h$ with significantly degenerated upper-bound that does not allow any refugees in the neighbourhood of~$h$. 
We show in the following auxiliary lemma that such inhabitants can be safely removed while preserving the (non-)existence of an inhabitants-respecting housing.

\begin{definition}
	An inhabitant~$h\in\I$ is called \emph{intolerant} if~$\ub(h) = 0$.
\end{definition}

\begin{lemma}
	\label{lem:remove-intolerant}
	Let~$\mathcal{I} = (G,\I,\R,\assgn,\ub)$ be an instance of the \ARHshort problem,~$h\in\I$ be an intolerant inhabitant, and~$F_j = \{\assgn(h)\}\cup U_h$.~$\mathcal{I}$ admits an inhabitants-respecting housing~$\pi$ if and only if the instance~$\mathcal{I}' = (G\setminus F_j,\I\setminus\{h\},R,\assgn,\ub)$ admits an inhabitants-respecting housing.
\end{lemma}
\begin{proof}
	Assume that $\mathcal{I}$ is a \YesI and $\pi$ is an inhabitants-respecting housing. Since~$h$ is intolerant, we have $\pi\cap U_h = \emptyset$. Therefore, $\pi$ is also a solution for $\mathcal{I}'$, as we only removed the houses in $F_h$ and the inhabitant~$h$.
	
	In the opposite direction, let $\mathcal{I}'$ be a \YesI and $\pi'$ be an inhabitants-respecting housing. As all refugees are housed on the vertices in $V(G)\setminus F_h$, we indeed have that $\pi'$ is a solution for $\mathcal{I}$.
\end{proof}

For the remainder of the paper, we assume that the input instance is reduced with respect to \Cref{lem:remove-intolerant}, that is, we assume only instances without intolerant inhabitants. Next, we show that the presence of intolerant inhabitants is not the only obstacle to guarantee the existence of inhabitants-respecting refugee housing.

\begin{proposition}
	\label{lem:housing-non-existence}
	There is an instance of the \ARHshort problem with no inhabitants-respecting refugees housing even if the instance does not contain intolerant agents.
\end{proposition}
\begin{proof}
	Let~$\ell\in\N$ be a natural number such that~$\ell \geq 2$. Consider an instance where the topology is complete bipartite graph~$G=(A\cup B,E)$, where~$A=\{a_1,a_2\}$,~$B=\{b_1,\ldots,b_\ell\}$, and~$E = \{\{a_i,b_j\}\mid i\in[2]\land j \in [\ell]\}$. There are two inhabitants~$h_1,h_2$ such that~$\ub(h_1) = \ub(h_2) = 1$ and for every~$i\in\{1,2\}$ we have~$\assgn(h_i) = a_i$, that is, the inhabitants occupy one side of the topology. Finally, suppose that~$2 \leq |\R| \leq \ell$. It is easy to see that we can house at most one refugee in the common neighbourhood of inhabitants~$h_1$ and~$h_2$. However,~${N_G(\assgn(h_1)) = N_G(\assgn(h_2)) = B = V_U}$, that is, all empty vertices are in their common neighbourhood and we have at least two refugees to house. Therefore, such an instance does not admit any inhabitants-respecting refugee housing.
\end{proof}

Naturally, the topology used in the proof of \Cref{lem:housing-non-existence} can be extended in multiple ways. In particular, we can arbitrarily increase the number of inhabitants and even their upper-bound, as long as the number of refugees and the size of part~$B$ is greater than~$\max_{h\in\I} \ub(h)$.

In the proof of \Cref{lem:housing-non-existence}, we use a construction where all inhabitants have the same upper-bound and the number of refugees is only greater by one compared to this upper-bound. In the following proposition, we introduce the only inhabitants-respecting housing existence guarantee we can give based on the upper-bound function and the number of refugees.

\begin{proposition}
	Let~$\mathcal{I} = (G,\I,\R,\assgn,\ub)$ be an instance of the\ARH problem and let ${\min_{h\in\I} \ub(h) \geq \R}$. Than~$\mathcal{I}$ is a \YesI.
\end{proposition}
\begin{proof}
	Let the housing $\pi$ be an arbitrary $\R$-sized subset of~$V$. Since $\min_{h\in\I} \ub(h)$ is greater than or equal to $\R$, $\pi$ is necessarily inhabitants-respecting.
\end{proof}

As the final result of this section, we give a simple observation that bounds the upper-bounds in terms of the number of refugees. In fact, we show that all inhabitants with upper-bound greater than or equal to the number of refugees to house can be safely removed.

\begin{lemma}
	Let~$\mathcal{I} = (G,\I,\R,\assgn,\ub)$ be an instance of the \ARHshort problem and $h\in\I$ be an inhabitant with~$\ub(h) \geq \R$.~$\mathcal{I}$ admits an inhabitants-respecting housing~$\pi$ if and only if the instance~$\mathcal{I}' = (G\setminus \{\assgn(h)\},\I\setminus\{h\},R,\assgn,\ub)$ admits an inhabitants-respecting housing.
\end{lemma}
\begin{proof}
	Let $\pi$ be an inhabitants-respecting housing in $\mathcal{I}$. Clearly, $\pi$ is also inhabitants-respecting in $\mathcal{I}'$ since the only change was the removal of $h$ and all remaining inhabitants were satisfied with $\pi$.
	In the opposite direction, let $\pi'$ be an inhabitants-respecting housing in $\mathcal{I}'$. If we return $h$ to $G'$, then we cannot exceed $h$'s upper-bound as $\ub(h) \geq \R = |\pi'|$.
\end{proof}

\section{Finding Inhabitant-Respecting Housings}

In the previous section, we showed that the existence of an inhabitants-respecting housing for a given instance of the \ARH problem is not guaranteed in general. In this section, we propose a number of algorithms that, given an instance, decide whether an inhabitants-respecting refugee housing exists and, if so, then return such housing~$\plc$. We support these algorithmic upper-bounds with respective hardness and running-time lower-bounds.

As our first result, we show that the \ARH problem is \NPc even if the topology is of bounded degree (and therefore also with bounded upper-bounds).

\begin{theorem}
	\label{thm:NPh}
	The \ARHshort problem is \NPc even if the topology is a planar graph of maximum degree~$3$.
\end{theorem}
\begin{proof}
	We give a reduction from the \textsc{Independent Set} problem which is known to be \NPh even on planar graphs of maximum degree three~\cite{GareyJ1977}. Recall that instance $\mathcal{I}$ of the \textsc{Independent Set} problem consists of a graph $G$ and an integer $k\in\N$, and the goal is to decide whether there exists a set $S\subseteq V(G)$ of size at least $k$ such that $G[S]$ is an edgeless graph.
	
	Given $\mathcal{I}=(G,k)$, we construct an equivalent instance $\mathcal{J}=(G',\I,\R,\assgn,\ub)$ of the \ARHshort problem as follows. Fix an ordering $(e_1,\ldots,e_m)$ of the edges of $G$. Now, we derive the graph $G'$ from $G$ by subdividing every edge $e_i$ with an intermediate vertex $w_i$. Every vertex $w_i$, $i\in[m]$, is occupied by an inhabitant~$h_i$ with $\ub(h_i) = 1$. These inhabitants ensure that the set of vertices occupied by refugees forms an independent set in the original graph. Observe that $V_U = V(G)$ and the subdivision of edges do not increase the maximum degree of the graph. Hence,~$G'$ is clearly of maximum degree three. Additionally, it is well-known that a graph is planar if and only if its subdivision is planar. Therefore, $G'$ is a planar graph. To complete the construction, we set $R=k$.
	
	For the correctness, let $\mathcal{I}=(G,k)$ be a \YesI and $S\subseteq V(G)$ be an independent set of size at least $k$. We fix an arbitrary ordering of $S$ and let $v_1,\ldots,v_k$ be the first $k$ vertices of this ordering. We set~${\pi=\{v_i\mid i\in[k]\}}$. Suppose that $\pi$ is not inhabitants-respecting. Then, there exists an inhabitant~$h_j$, ${j\in[m]}$, such that $|N_G(\assgn(h_j)) \cap \pi| = 2$. Let $u,w \in N_G(\assgn(h_j))\cap \pi$ be two distinct vertices in the neighbourhood of~$h_j$ with refugees housed on them. Then, by construction, $\{u,v\}\in E(G)$, which contradicts $S$ being an independent set.
	
	In the opposite direction, let the instance $\mathcal{J}$ be a \YesI and $\pi$ be an inhabitants-respecting housing. We set $S = \pi$ and, for the sake of contradiction, assume that $S$ is not an independent set in $G$. Then there exists a pair of distinct vertices $u,v\in S$ such that $\{u,v\}\in E(G)$. Consequently, $u$ and $v$ are also members of $\pi$ and since in $G'$ the edge $\{u,v\}$ is subdivided and there is a vertex~$w$ occupied by an inhabitant $h$ with $\ub(h) = 1$, it follows that $\pi$ is not inhabitants-respecting due to $h$. This is a contradiction, and $S$ is clearly an independent set in $G$ of size $k$.
	
	Finally, it is easy to see that the reduction can be done in polynomial time. Additionally, given a set $\pi$, it is easy to verify in polynomial time whether $\pi$ is inhabitants-respecting by enumerating all inhabitants and comparing their upper-bounds with the number of refugees assigned to their neighbourhood. Hence, the problem is clearly in \NP, finishing the proof.
\end{proof}

In contrast to \Cref{thm:NPh}, Knop and Schierreich~\cite{KnopS2023} showed that if the topology is of maximum degree ($\Delta$) two, then the problem can be solved in cubic time even if the inhabitants approve arbitrary intervals. It directly follows that the polynomial-time solvability holds even in our special case with upper-bounds; however, as we will show later, there exists a much simpler and faster algorithm for our setting.

We begin our algorithmic journey with a polynomial-time algorithm that decides the existence of an inhabitants-respecting housing when the topology is a forest.

\begin{theorem}\label{thm:forest:poly}
	Every instance of \ARH where the topology is a forest can be solved in $\Oh{n\cdot(\Delta+\Delta\log \Delta)}$ time, where $\Delta = \max_{v\in V(G)}\deg(v)$.
\end{theorem}
\begin{proof}
	\newcommand{\DP}{\ensuremath{\operatorname{DP}}}
	Recall that due to \Cref{lem:remove-intolerant} we may assume that the instance does not contain intolerant inhabitants. Our algorithm is based on the leaf-to-root dynamic programming approach, and we run the algorithm for every connected component of the topology separately.
	
	Let $T$ be a tree rooted in a vertex $\operatorname{r}(T)$. In the dynamic programming table~$\DP$, we store for every vertex $v\in V(G)$ two values $\DP_v[ 0 ]$ and $\DP_v[ 1 ]$. The semantics of the keys differs based on the type of vertex: \begin{itemize}
		\item If $v\in V_U$, then in $\DP_v[ 0 ]$ we store the maximum number of refugees we can house in the sub-tree rooted in $v$ assuming that $v$ remains unoccupied in the solution, and in $\DP_v[ 1 ]$ we store the maximum number of refugees we can house in the sub-tree rooted in $v$ assuming that there is a refugee housed in $v$.
		\item If $v$ is occupied by an inhabitant $h$, then $\DP_v[ 0 ]$ stores the maximum number of refugees we can house in the sub-tree rooted in $v$ assuming that the parent of $v$ is empty, and $\DP_v[ 1 ]$ stores the maximum number of refugees we can house in the sub-tree rooted in~$v$ assuming that a refugee is housed in the parent of $v$, respectively.
	\end{itemize}
	
	The computation runs from leaves to root as follows. We describe the procedure separately for leaves and internal vertices and prove the correctness of every step by induction.
	
	First, let $v$ be a leaf, and $v\in V_U$. Then, we set ${\DP_v[ 0 ] = 0}$ and $\DP_v[ 1 ] = 1$. Observe that this is clearly correct, as the sub-tree rooted in $v$ can house at most $1$ refugee and the value is $1$ if and only if the vertex $v$ is assumed to be occupied in the solution.
	
	If $v \notin V_U$ and $v$ is a leaf, we set $\DP_v[ 0 ] = \DP_v[ 1 ] = 0$. Again, it is easy to see that this computation is correct as the sub-tree rooted in $v$ contains no empty house and therefore we cannot house any refugee in this sub-tree regardless of occupancy of the parent.
	
	Next, let $v\in V_U$ be an internal vertex. For $i\in\{0,1\}$ we set $\DP_v[ i ] = \sum_{w\in\operatorname{children}(v)} \DP_w[ i ] + i$. The occupancy of $v$ is prescribed by the value of $i$ and the maximum number of refugees that we can house in the sub-tree rooted in $v$ is clearly the sum of refugees we can house in the sub-trees rooted in children of $v$ plus $1$ if $v$ is also used for housing. By the induction hypothesis, these values are computed correctly and the correctness for this type of vertices follows. 
	
	Finally, assume that $v$ is an internal vertex occupied by an inhabitant. We start the computation by creating a max-heap $H$ and inserting the value $g_w = \DP_w[ 1 ] - \DP_w[ 0 ]$ for every $w\in\operatorname{children}(v)$ such that the value $g_w$ is strictly positive. The idea is to have information about the additional value we can gain from every sub-tree rooted in a child $w$ if we allow $w$ to be also occupied. Next, we set for every $i\in\{0,1\}$ the value of $\DP_v[ i ] = \sum_{u\in\operatorname{children}(v)}\DP_u[ 0 ]$ and we additionally add the $\min\{\ub(v) - i,\operatorname{size}(H)\}$ highest values stored in $H$. That is, the number of refugees we can house in the sub-tree rooted in $v$ is the sum of at most $\ub(\assgn^{-1}(v)) - i$ sub-trees whose root is occupied, and the rest of sub-trees roots have to be unoccupied to not exceed the upper-bound of the inhabitant occupying $v$. By using the max-heap, we ensure that we indeed obtain the maximum possible value.
	
	Let $T_1,\ldots,T_k$ be connected components of $G$. Since $G$ is a forest, each component is a tree. Therefore, we use the preceding algorithm to fill the dynamic programming tables for all $\operatorname{r}(T_i)$, $i\in[k]$. Now, if $\sum_{i\in[k]}\max_{j\in\{0,1\}}\DP_{\operatorname{r}(T_i)}[j]$ is at least $\R$, then we return \Yes. Otherwise, we return \No. The computation is constant in leaves, takes linear time (with respect to the maximum degree) in empty internal vertices, and $\Oh{\Delta\cdot\log \Delta}$ in occupied internal vertices due to the heap, finishing the proof.
\end{proof}

With \Cref{thm:forest:poly} in hand, we can finally give an improved algorithm for graphs of maximum degree two. Recall that this provides a clear dichotomy between tractable and intractable cases with respect to the maximum degree of topology.

\begin{theorem}
	\label{thm:maxdeg2:poly}
	Every instance of \ARHshort where the topology is a graph with $\Delta \leq 2$ can be solved in $\Oh{n}$ time.
\end{theorem}
\begin{proof}
	First, observe that every bipartite graph of maximum degree $2$ is a collection of paths and even cycles~\cite{Diestel2017}. We assume that the topology is a bipartite graph, the instance is reduced with respect to \Cref{lem:remove-intolerant}, and let $\mathcal{C} = \{C_1,\ldots,C_\ell\}$ be a set of all connected components of~$G$.
	
	For every component $C_i$, $i\in[\ell]$, such that $C_i$ is isomorphic to a path, we can use the algorithm from \Cref{thm:forest:poly} to compute the maximum number $R_{C_i}$ of refugees that can be housed on vertices of~$C_i$ in $\Oh{n}$ time, since $\Delta = 2$.
	
	Next, we deal with components that are isomorphic to even cycles. Let $C\in\mathcal{C}$ be such a component. If all inhabitants assigned to $C$ have upper-bound equal to~$2$, then all empty vertices of $C$ can be occupied by refugees, and therefore we set $R_C = |V(C) \cap V_U|$. Otherwise, there is at least one inhabitant $h$ with upper-bound equal to $1$. Let $v$ be the vertex to which the inhabitant~$h$ is assigned, and let $\{u,w\} = N(v)$. Now, there are three possible shapes of $h$'s neighbourhood in hypothetical solution $\pi$:
	\begin{enumerate}
		\item neither $u$ nor $w$ are part of $\pi$. In this case, we can remove $N[v]$ from $C$ and compute the value $R^1_c$, which is the maximum number of refugees we can house on the remaining path using the algorithm from \Cref{thm:forest:poly} in linear time.
		\item $u$ is part of $\pi$, while $w$ is not. We again remove $N[v]$ from~$C$. Additionally, we decrease the upper-bound of the vertex $x \in N(u)\setminus\{v\}$ by one and apply \Cref{lem:remove-intolerant} as this step can create an intolerant inhabitant. What remains is a path; we use the algorithm from \Cref{thm:forest:poly} to compute $R'_C$ and set $R^2_C = 1 + R'_C$.
		\item $w$ is part of $\pi$, while $u$ is not. We compute $R^3_C$ using the symmetric procedure as in the previous case.
	\end{enumerate}
	Finally, we set $R_C = \max\{R^1_C,R^2_C,R^3_C\}$. As every case can be computed in linear time, we see that computation for even cycles can also be done in linear time.
	As our last step, we check whether $\sum_{C\in\mathcal{C}} R_C \geq \R$. If that is the case, we are dealing with \YesI, otherwise, the result is \No.
\end{proof}

\subsection{Natural Parameters}

In this section, we turn our attention to the natural parameters of the problem. In particular, we study the complexity of the problem when there is a small number of refugees, inhabitants, or empty vertices.

The first question arising in this direction is about parameterized complexity with respect to the number of refugees. Knop and Schierreich~\cite{KnopS2023} showed that the general case of anonymous refugee housing is \Wh[2] with respect to this parameter. Is our restriction strong enough to make the problem tractable? Unfortunately not, as we show in the following result.

\begin{theorem}
	\label{thm:R:Whard}
	The \ARHshort problem is \Wh parameterised by the number of refugees~$R$, there is an algorithm deciding the \ARHshort problem in~$n^\Oh{\R}$ time%
	, and, unless ETH fails, there is no algorithm with running time~$f(\R)\cdot n^{o(\R)}$ for any computable function~$f$.
\end{theorem}
\begin{proof}
	The \Whness directly follows from the proof of \Cref{thm:NPh}, as we used~$R=k$ and the reduction was polynomial time.
	
	Next, the algorithm is a simple brute-force. There are~$\binom{V_U}{\R}\in n^\Oh{\R}$ possible~$R$-sized subsets of empty houses and for every such subset~$\pi'$, we can check in polynomial-time whether~$\pi'$ is inhabitants-respecting. Since we tried all possible solutions, the algorithm is indeed correct.
	
	Finally, we show the running time lower-bound. It is known that the \textsc{Independent Set} problem cannot be solved in $f(k)\cdot n^{o(k)}$ time unless the ETH fails~\cite{ChenHKX2006}. For the sake of contradiction, suppose that there is an algorithm~$\mathbb{A}$ solving \ARHshort in $f(\R)\cdot n^{o(\R)}$ time. Suppose the following algorithm for \textsc{Independent Set}. Given an instance~$\mathcal{I}$, we use the construction of the proof of \Cref{thm:NPh} to obtain an equivalent instance~$\mathcal{J}$ of \ARHshort. Then we solve $\mathcal{J}$ using the algorithm~$\mathbb{A}$ in $f(R)\cdot n^{o(R)}$ time, and based on the response of $\mathbb{A}$, we decide $\mathcal{I}$. As the construction of $\mathcal{J}$ runs in polynomial time and $R=k$, this gives us an algorithm for the \textsc{Independent Set} problem running in $f(k)\cdot n^{o(k)}$ time. This contradicts ETH, which is unlikely.
\end{proof}

The next result, which is of great importance to us, is an \FPT algorithm with respect to the number of inhabitants. It follows directly from the ILP formulation of the more general model of anonymous refugee housing; however, we use it as a sub-procedure of our algorithms in future sections, and therefore we state it formally.

\begin{lemma}[Knop and Schierreich~\cite{KnopS2023}]\label{lem:ARH:FPT:I}
	The \ARHshort problem is fixed-parameter tractable when parameterised by the number of inhabitants~$|\I|$.
\end{lemma}

The next natural parameter that arises in the context of the \ARHshort problem is the number of empty vertices. For this parameterisation, the problem can be solved in \FPT time using a simple $2^\Oh{|V_U|}$ algorithm as shown by~\cite{KnopS2023}. Moreover, this algorithm is best possible assuming the Exponential-Time Hypothesis.

As the final result of this section, we identified a new natural parameter $\xi$, which is the number of empty vertices minus the number of refugees to house. It is easy to see that in our model, if the number of refugees is equal to the number of empty vertices (and therefore $\xi=0$), then the problem can be decided in polynomial time, as there is only one possible housing, and checking whether this housing is inhabitants-respecting is also polynomial time. 

\begin{theorem}
	\label{thm:remaining:XP}
	The \ARHshort problem is in \XP when parameterised by the number of extra empty houses $\xi = |V_U| - \R$.
\end{theorem}
\begin{proof}
	The algorithm is a simple brute-force. We guess a set of empty houses $F\subseteq V_U$ and for each such set $F$, we create an instance with a topology $G\setminus F$. For this new topology, we have $|V_U| = \R$ and, therefore, there is only one solution that can be verified in polynomial time. Note that there are $n^\Oh{\xi}$ such sets~$F$, which concludes the proof.
\end{proof}

We find the question of whether the algorithm from \Cref{thm:remaining:XP} can be turned into an \FPT one very intriguing, but all our attempts to achieve this have so far been unsuccessful.

\subsection{Restricting the Topology}

Now, we turn our attention to different restrictions of the topology. Here, we follow the usual approach of distance to triviality~\cite{GuoHN2004,Niedermeier2010,AgrawalR2022} -- we identify polynomial-time solvable graph families and then turn our attention to structural restrictions that are ``not far'' from these trivial cases.

\subsubsection{Tree-Like Networks.}

\noindent The \Cref{thm:forest:poly} shows that \ARHshort can be decided efficiently if the topology is a forest. It is natural to ask whether this algorithm can be generalised to a larger class of graphs that are similar to forests. Natural candidates are graphs of bounded tree-width, since a graph has tree-width equal to $1$ if and only if the graph is isomorphic to a forest~\cite{RobertsonS1986}. We resolve this question in a partly positive way in the following theorem.

\begin{theorem}
	\label{thm:tw:XP}
	The \ARHshort problem is in \XP when parameterised by the tree-width $\tw$ of the topology.
\end{theorem}
\begin{proof}
	\newcommand{\DP}{\ensuremath{\operatorname{DP}}}
	As is usual for algorithms on bounded tree-width graphs, the algorithm is leaf-to-root dynamic programming along a nice tree decomposition. An optimal tree decomposition of width $k$ can be found in $2^\Oh{k^2}\cdot n^\Oh{1}$ time using the algorithm of~Korhonen and Lokshtanov~\cite{KorhonenL2023}.
	
	In every node $x$ of the tree decomposition, we store a dynamic programming table $\DP_x[A,P,F,\rho]$ (we use $\beta(x)$ to denote the set of vertices associated with node $x$), where
	\begin{itemize}
		\item $A\subseteq V_U\cap\beta(x)$ is the set of empty vertices used in the solution,
		\item $P\colon V_I\cap\beta(x)\to[n]_0$ is the number of already forgotten (or past) neighbours of each $v\in V_\I\cap\beta(x)$ used to house refugees in the solution,
		\item $F\colon V_I\cap\beta(x)\to[n]_0$ is the number for every $v\in V_\I\cap\beta(x)$ of not yet observed neighbours that will be used to house refugees in the solution, and
		\item $\rho\in[\R]_0$ is the number of refugees we want to house in this particular sub-tree.
	\end{itemize}
	We call the quadruple $(A,P,F,\rho)$ a \emph{signature}. The value stored in the dynamic programming table $\DP_x$ for every signature $(A,P,F,\rho)$ is \texttt{true} if and only if there exists an inhabitants-respecting \emph{partial housing} $\pi_x\subseteq V^x\cap V_U$ such that $|\pi_x| = \rho$, $\pi_x \cap \beta(x) = A$, and for every $w\in \beta(x)\cap V_\I$ we have $|(\pi_x \cap N(w))\setminus A| = P(w)$ and $P(w) + F(w) + |A\cap N(w)| \leq \ub(\assgn^{-1}(w))$. If such a partial housing does not exist, we store \texttt{false}. Once the dynamic programming table $\DP_r$ for the root node $r$ is computed correctly, we simply check whether $\DP_r[\emptyset,\emptyset,\emptyset,R]$ is \texttt{true}. 
	
	Now, we describe how the algorithm proceeds. We describe the computation separately for each node type of the nice-tree decomposition.
	
	\paragraph{Leaf Node.} Let $x$ be a leaf node. By definition, the leaf nodes are empty and, therefore, we cannot house any refugees in the subtree rooted in $x$. Consequently, we set
	\[
	\DP_x[A,P,F,\rho] = \begin{cases}
		\texttt{true} & \text{if } \rho = 0,\\
		\texttt{false} & \text{otherwise.}
	\end{cases}
	\]
	
	\begin{lemma}\label{claim:tw:leaf}
		The dynamic programming for every leaf node $x$ is filled correctly.
	\end{lemma}
	\begin{claimproof}
		We show that the dynamic programming table $\DP_x$ stores \texttt{true} if and only if there is a partial housing $\pi_x$ compatible with the given signature $(A,P,F,\rho)$. For the left-right direction, assume that $\DP_x[A,P,F,\rho]$ is \texttt{true}. By the definition of the computation, this can happen only if the $\rho = 0$, that is, there is no refugee to house. Consequently, an empty housing $\pi_x$ is clearly compatible with $(A,P,F,\rho)$. In the opposite direction, let there be a partial housing compatible with some signature $(A,P,F,\rho)$. Since $G^x$ is an empty graph, the housing is clearly empty. Hence, $\rho = 0$. However, in such a case, we clearly store \texttt{true} in the cell $\DP_x[A,P,F,\rho]$.
	\end{claimproof}
	
	\paragraph{Introduce Empty Vertex Node.} Let $x$ be an introduce node that introduces an empty node $v$. If $v\not\in A$, then this vertex is irrelevant to the solution, and we can remove it completely from the topology. Otherwise, for every neighbour of $v$, a new vertex of future houses moved to actual houses. Therefore, we copy the value of $\DP_y$, where the value of $F(w)$ for every $w\in N(v)$ is higher by $1$ compared to the guessed value of $F(w)$. This yields the following:
	\[
	\DP_x[A,P,F,\rho] = \begin{cases}
		\DP_y[A,P,F,\rho] & \text{if } v\not\in A,\\
		\texttt{false} & \text{if } \exists w\in N(v)\cap\beta(x)\colon F(w) = n,\\
		\texttt{false} & \text{if } \rho = 0\text{, and} \\
		\DP_y[A\setminus\{v\},P,F',\rho-1] & \text{otherwise,}
	\end{cases}
	\]
	where $\forall w\in\beta(x)\cap V_\I$ we have $F'(w) = F(w)$ if $w\not\in N(v)$ and $F'(w) = F(w) + 1$ if $w\in N(v)$. Note that the second and third cases are distinguished just for technical reasons, as without them, the signature we use to question $y$'s table in case four would be invalid.
	
	\begin{lemma}\label{claim:tw:introduceEmpty}
		Let $x$ be an introduce empty vertex node with a single child node $y$. Assuming that the dynamic programming table $\DP_y$ is computed correctly, dynamic programming $\DP_x$ is also correct.
	\end{lemma}
	\begin{claimproof}
		Let $\DP_x[A,P,F,\rho]$ be \texttt{true}. Assume first that $v\not\in A$. Then, according to the definition of the computation, $\DP_y[A,P,F,\rho]$ is \texttt{true}. According to the induction hypothesis, there exists a partial housing $\pi_y$ compatible with $(A,P,F,\rho)$. Since the signature is the same, $\pi_y$ is the sought partial housing. Now, let $v\in A$. Since $\DP_x[A,P,F,\rho]$ is \texttt{true}, it holds that $\DP_y[A\setminus\{v\},P,F',\rho-1]$ is \texttt{true}. By the induction hypothesis, there exists a partial housing $\pi_y$ compatible with $(A\setminus\{v\},P,F',\rho-1)$. We set $\pi_x = \pi_y \cup \{v\}$ and claim that it is a partial solution compatible with $(A,P,F,\rho)$. Clearly, we increased the housing size by exactly one. The inhabitants occupying the non-neighbours of $v$ are still satisfied with their neighbourhoods, as housing a refugee in $v$ does not affect them. Let $w\in N(v)$ be a neighbour of $v$. Since $\pi_y$ is a partial housing compatible with $(A\setminus\{v\},P,F',\rho-1)$, it holds that $\ub(\assgn^{-1}(w)) \geq P(w) + F'(w) + |(A)\setminus\{v\} \cap N(w)|$. Observe that in $\pi_x$, the $\assgn^{-1}(w)$'s upper-bound is not exceeded, as it has at most $P(w) + F(w) + |A\cap N(w)| = P(w) + F'(w) - 1 + |(A\setminus\{v\})\cap N(w)| + 1 = P(w) + F'(w) + |(A\setminus\{v\}) \cap N(w)|$ refugees in the neighbourhood. Therefore, $\pi_x$ is indeed a partial solution compatible with $(A,P,F,\rho)$.
		
		In the opposite direction, let there be a partial housing $\pi_x$ compatible with some signature $(A,P,F,\rho)$. We again distinguish two cases based on the presence of $v$ in $A$. First, suppose that $v\not\in A$. Observe that since the partial housing does not use $v$, then it remains also a partial housing compatible with $(A,P,F,\rho)$ on a graph $G^x\setminus\{v\} = G^y$. By the induction hypothesis, $\pi^x$ is also compatible with $(A,P,F,\rho)$ in $y$ and consequently, $\DP_y[A,P,F,\rho]$ is set to \texttt{true}. By the definition of the computation, also $\DP_x[A,P,F,\rho]$ is set to \texttt{true} since it only takes over the value stored in $\DP_y[A,P,F,\rho]$ in this case. It remains to show the case where $v\in A$. Since $\pi_x$ is a partial housing, every inhabitant of $G^x$ approves the number of refugees housed in its neighbourhood. Specifically, for each inhabitant $v\in V_I \cap \beta(x)$, it holds that $P(v) + F(v) + |A| \leq \ub(\assgn^{-1}(v))$. If we remove $v$ from the graph, we clearly obtain the graph $G^y$. By setting $\pi_y = \pi_x \setminus \{v\}$, we clearly obtain a partial housing in $y$: The non-neighbours of $v$ are the same, and for the neighbours of $v$, we only increased the number of refugees in their neighbourhoods. Moreover, $\pi_y$ is clearly compatible with the signature $(A\setminus\{v\},P,F',\rho)$, where $F'(w) = F(w)$ for every $w\not\in N(v)$ and $F'(w) = F(w) + 1$ otherwise. We assumed that the dynamic programming table for the child node $y$ is computed correctly and, therefore, $\DP_y[A\setminus\{v\},P,F',\rho]$ is \texttt{true}. By the definition of the computation, also $\DP_x[A,P,F',\rho]$ is \texttt{true}, as we are in the fourth case of the recurrence.
	\end{claimproof}
	
	\paragraph{Introduce Occupied Vertex Node.} If $x$ is an introduce node introducing an occupied vertex $v$, then clearly, the stored value should be \texttt{false} whenever $P(v) > 0$, as $v$ cannot have any already forgotten neighbours in the solution by the definition of tree decomposition. The stored value is also \texttt{false} if the number of neighbours of $v$ in $A$ and the number of future refugees in the neighbourhood of~$v$ exceeds the upper-bound of an inhabitant occupying $v$. In all the remaining cases, $v$ is satisfied in the solution, and hence, we need to check the stability of the other inhabitants in the bag. However, this can be read from the dynamic programming table of node $y$. Thus, we get the following computation:
	
	\[
	\DP_x[A,P,F,\rho] = \begin{cases}
		\texttt{false} & \text{if } P(v) \not= 0\\
		\texttt{false} & \text{if } |A\cap N(v)|+F(v) > \ub(\assgn^{-1}(v))\\
		\DP_y[A,P_{\downarrow v},F_{\downarrow v},\rho] & \text{otherwise,}
	\end{cases}
	\]
	where $P_{\downarrow v}$ and $F_{\downarrow v}$ means that $v$ is excluded from the domain of the respective function.
	
	\begin{lemma}\label{claim:tw:introduceOccupied}
		Let $x$ be an introduce an occupied vertex node with child $y$. Assuming that the dynamic programming table $\DP_y$ is computed correctly, the dynamic programming table for $x$ is also correct.
	\end{lemma}
	\begin{claimproof}
		First, let $\DP_x[A,P,F,\rho]$ be $\mathtt{true}$. By the definition of the computation, this holds if and only if $P(v) = 0$, $|A\cap N(v)| + F(v) \leq \ub(\assgn^{-1}(v))$, and $\DP_y[A,P_{\downarrow v},F_{\downarrow v},\rho] = \mathtt{true}$. By assumption, the dynamic programming table $\DP_y[A,P_{\downarrow v},F_{\downarrow v},\rho]$ is computed correctly and, therefore, there exists a partial housing $\pi_y$ compatible with the signature $(A,P_{\downarrow v},F_{\downarrow v},\rho)$. We set $\pi_x = \pi_y$ and claim that it is indeed a partial housing compatible with $(A,P,F,\rho)$. Clearly, by the definition of tree decomposition, there is no neighbour of $v$, which is already forgotten in $G^x$. Hence, $P(v) = 0$ holds. Next, the number of refugees already assigned to the neighbourhood of $\assgn^{-1}(v)$ and the number of refugees that will be housed in its neighbourhood in the future is at most $\assgn^{-1}(v)$'s upper-bound, otherwise the stored value would be \texttt{false}. Hence, the partial housing~$\pi_x$ is inhabitants-respecting from the perspective of the inhabitant assigned to $v$. As it was also inhabitants-respecting from the perspective of all the remaining inhabitants, the correctness of $\pi_x$ follows.
		
		In the opposite direction, let $(A,P,F,\rho)$ be a signature and $\pi_x$ be a corresponding partial housing. Since $\pi_x$ is a partial housing, the number of refugees currently assigned to the neighbourhood of the inhabitant occupying $v$ plus the number of refugees expected to be assigned to the $\assgn^{-1}(v)$'s neighbourhood in the future is at most its upper-bound. Hence, we can clearly remove the vertex $v$ from the graph and obtain exactly the graph $G^y$. As all inhabitants of $G^y$ were satisfied in $\pi^x$, they remain the same if we remove $v$. Consequently, $\pi_x$ is also a partial housing corresponding to the signature $(A,P_{\downarrow v},F_{\downarrow v},\rho)$ in $y$. Since we assumed that the dynamic programming table $\DP_y$ is computed correctly, it necessarily follows that $\DP_y[A,P_{\downarrow v},F_{\downarrow v},\rho]$ is \texttt{true}. However, by the definition of the computation, $\DP_x[A,P,F,\rho]$ is also set to \texttt{true}.
	\end{claimproof}
	
	\paragraph{Forget Empty Vertex Node.} If we are forgetting an empty vertex, we need to decide whether $v$ can or cannot be part of the solution housing. This can be easily read from the dynamic programming table of the child node $y$.
	\[
	\DP_x[A,P,F,\rho] = \begin{cases}
		\DP_y[A,P,F,\rho] & \text{if } \exists w\in N(v)\cap\beta(x)\colon P(w) = 0,\\
		\DP_y[A,P,F,\rho] \lor \DP_y[A\cup\{v\},P',F,\rho] & \text{otherwise.}
	\end{cases} 
	\]
	where $\forall w\in \beta(x)\cap V_\I$ we set $P'(w) = P(w)$ if $w\not\in N(v)$ and $P'(w) = P(w) - 1$ if $w\in N(v)$.
	
	\begin{lemma}\label{claim:tw:forgetEmpty}
		Let $x$ be a forget an empty vertex node with child $y$. Assuming that the dynamic programming table $\DP_y$ is computed correctly, the dynamic programming table for $x$ is also correct.
	\end{lemma}
	\begin{claimproof}
		Let $\DP_x[A,P,F,\rho]$ be $\mathtt{true}$. Then, according to the definition of the computation, $\DP_y[A,P,F,\rho]$ or $\DP_y[A\cup\{v\},P',F,\rho]$, where $P'(w) = P(w)$ if $w\not\in N(v)$ and $P'(w) = P(w)-1$ otherwise, is \texttt{true}. First, assume that it is the first case. Then there exists a partial housing $\pi_y$ corresponding to the signature $(A,P,F,\rho)$ in $G^y$. Since $G^x = G^y$, $\pi_y$ is clearly also a solution housing corresponding to $(A,P,F,\rho)$ in $x$. To complete the implication, assume that $\DP_y[A,P,F,\rho] = \mathtt{false}$. Consequently, $\DP_y[A\cup\{v\},P',F,\rho]$ is \texttt{true} and by our assumption there exists a partial housing $\pi_y$ corresponding to $(A\cup\{v\},P',F,\rho)$ in $y$. We claim that $\pi_y$ is also a partial housing corresponding to $(A,P,F,\rho)$. Clearly, $|\pi_y| = \rho$ and $\pi_y \subseteq V^y \subseteq V^x$. It remains to show that the conditions for all the bag vertices are satisfied. For all non-neighbours of $v$, all the conditions are satisfied as they are not affected by the housing of $v$. So, let $w$ be a bag vertex that is a neighbour of $v$. We have $P'(w) + F(w) + |A\cup\{v\}\cap N(w)| \leq \ub(\assgn^{-1}(w))$, which can be rewritten as $P(w) - 1 + F(w) + |A| + 1 = P(w) + F(w) + |A\cap N(w)| \leq \ub(\assgn^{-1}(w))$. This shows that $\pi_y$ does not violate the upper-bound of $w$. Moreover, we have $|(\pi_y \cap N(w))\setminus (A\cup\{v\})| = P'(w)$. This can again be rewritten as $|(\pi_y \cap N(w))\setminus A| - 1 = P(w) - 1$, which implies $|(\pi_y \cap N(w))\setminus A| = P(w)$. Consequently, also the number of past vertices is correct and we have that $\pi_x = \pi_y$ is a partial housing compatible with $(A,P,F,\rho)$.
		
		In the opposite direction, assume that for a signature $(A,P,F,\rho)$, there is a partial housing $\pi_x$. First, assume that $v\not\in \pi_x$. By this, we immediately obtain that $\pi_x$ is also compatible with $(A,P,F,\rho)$ in $y$, as $G^y = G^x$ and the only difference between $x$ and $y$ is in the absence of $v$ in $\beta(x)$. By our assumption, the dynamic programming table $\DP_y$ contains \texttt{true} for the signature $(A,P,F,\rho)$, and therefore also $\DP_x[A,P,F,\rho]$ is \texttt{true}. Next, assume that $v\in \pi_x$. Since $\pi_x$ is correct, clearly $P(w) \geq 1$ for every neighbour $w\in\beta(x)$ of $v$. We claim that $\pi_x$ corresponds to the signature $(A\cup\{v\},P',F,\rho)$, where $P'(w) = P(w)$ if $w\not\in N(v)$ and $P'(w) = P(w)-1$ otherwise. Since $\beta(y) = \beta(x)\cup\{v\}$, clearly $\pi_x\cap\beta(y) = A\cup\{v\}$. The number of housed refugees is trivially correct. Let $w\in\beta(x)$ be a non-neighbour of $v$. We have $P(w) + F(w) + |A\cap N(w)| = P'(w) + F(w) + |(A\cup\{v\}) \cap N(w)| \leq \ub(\assgn^{-1}(w))$ and $|(\pi_x \cap N(w))\setminus A| = |(\pi_x \cap N(w))\setminus (A\cup\{v\})| = P(w) = P'(w)$. Therefore, for non-neighbours, the conditions hold. Let $w\in\beta(y)$ be a neighbour of $v$. We have $P(w) + F(w) + |A \cap N(w)| = P'(w) + 1 + F(w) + |(A\cup \{v\})\cap N(w)| - 1 = P'(w) + F(w) + |(A\cup \{v\})\cap N(w)| \leq \ub(\assgn^{-1}(w))$ and $|(\pi_x \cap N(w))\setminus A| = |(\pi_x \cap N(w))\setminus (A\cap\{v\})| + 1 = P(w) = P'(w) + 1$, and therefore $P'(w) = |(\pi_x \cap N(w))\setminus (A\cap\{v\})|$. Hence, $\pi_x$ is compatible with $(A\cup\{v\},P',F,\rho)$ in $y$. By assumption, $\DP_y[A\{v\},P',F,\rho]$ is \texttt{true} and consequently also $\DP_x[A,P,F,\rho]$ is \texttt{true} by the definition of the computation.
	\end{claimproof}
	
	\paragraph{Forget Occupied Vertex Node.} If an occupied vertex $v$ is forgotten in a node $x$, then we need to check how many inhabitants are, in a hypothetical solution, housed in the neighbourhood of $v$. Formally, we set
	\[
	\DP_x[A,P,F,\rho] = \bigvee_{i\in[n]_0}\DP_y[A,P_{\uparrow (v\mapsto i)},F_{\uparrow (v\mapsto 0)},\rho],
	\]
	where $F_{\uparrow (v\mapsto 0)}$ means that the function $F$ is extended such that for~$v$ it returns $0$. We use the same notation for $P_{\uparrow (v\mapsto i)}$.
	
	\begin{lemma}\label{claim:tw:forgetOccupied}
		Let $x$ be a forget an occupied vertex node with the child $y$. Assuming that the dynamic programming table $\DP_y$ is computed correctly, the dynamic programming table for $x$ is also correct.
	\end{lemma}
	\begin{claimproof}
		Assume that $\DP_x[A,P,F,\rho]$ is $\mathtt{true}$. By the definition of the computation, there exists $i\in[n]_0$ such that $\DP_y[A,P'=P_{\uparrow(v\mapsto i)}, F'=F_{\uparrow(v\mapsto 0)},\rho]$ is \texttt{true}. By assumption, there exists a partial housing $\pi_y$ compatible with $(A,P_{\uparrow(v\mapsto i)}, F_{\uparrow(v\mapsto 0)},\rho)$ in $y$. We claim that $\pi_y$ is also compatible with $(A,P,F,\rho)$. Recall that $\beta(x) = \beta(y)\setminus\{v\}$, $G^x = G^y$, and that the only difference between $P'$ and $P$ (and $F'$ and $F$, respectively) is that the vertex $v$ is removed from the domain of $P'$. So, the values for all $w\in\beta(x)$ are the same in both $P'$ and $P$ ($F'$ and $F$) and therefore it is obvious that $\pi_y$ is also compatible with $(A,P,F,\rho)$.
		
		In the opposite direction, let $(A,P,F,\rho)$ be a signature and $\pi_x$ be a corresponding partial housing. Since $G^x = G^y$ and $\beta(y) = \beta(x)\cup\{v\}$, $\pi_x$ is also a partial housing for some signature $(A,P',F',\rho)$ in $y$. Since $v$ is forgotten in $x$, there is no neighbour $w$ of $v$ such that $v\not\in V^x$. Hence, we can freely assume that $F'(v) = 0$ and $F'(u) = F(u)$ for all $u\in\beta(x)$. Let $X = \pi_x \cap N(v)$ be the set of $v$'s neighbours in the solution. Then for $P'(v) = |X\setminus A|$ and $P'(u) = P(u)$ for all $u\in\beta(x)$, all the conditions for $\pi_x$ being a partial housing compatible with $(A,P',F',\rho)$ are clearly satisfied. Hence, $\DP_y[A,P',F',\rho]$ is necessarily \texttt{true} and therefore, $\DP_x[A,P,F,\rho]$ is also \texttt{true} by the recurrence.
	\end{claimproof}
	
	\paragraph{Join Node.} A join node $x$ is a node with two children $y$ such that $\beta(x) = \beta(y) = \beta(z)$. We set the value of a cell to \texttt{true} if and only if there exists a pair of cells (set to \texttt{true}) in the children $y$ and $z$ such that
	\begin{itemize}
		\item the set of empty houses assumed to be in the solutions is the same for both children,
		\item the number of forgotten vertices in $y$ and in $z$ for every $v\in \beta(x)\cap V_\I$ sum up to $P(v)$,
		\item the functions $F_y$ and $F_z$ reflect that some future vertices come from the other sub-tree and some should be expected in the future, and
		\item the sum of refugees minus $|A|$ (as it is counted twice -- once in $y$ and once in $z$) is equal to the guessed $\rho$.
	\end{itemize}
	Formally, the computation is as follows:
	\[
	\DP_x[ A, P, F, \rho ] = 
	\bigvee_{\mathclap{\substack{
				\rho_y,\rho_z\colon \rho_y + \rho_z = \rho+|A|\\
				P_y,P_z\colon \forall v\in\beta(x)\cap V_\I\colon P_y(v)+P_z(v) = P(v)\\
				F_y,F_z\colon \forall v\in\beta(x)\cap V_\I\colon F_y(v) = F(v) + P_z(v)\,\land\, F_z(v) = F(v) + P_y(v)
	}}}
	\left(\DP_y[A,P_y,F_y,\rho_y]\,\land\,\DP_z[A,P_z,F_z,\rho_z]\right).
	\]
	
	\begin{lemma}\label{claim:tw:join}
		Let $x$ be a join node with two children $y$ and $z$. Assuming that the dynamic programming tables $\DP_y$ and $\DP_z$ are computed correctly, the dynamic programming table for $x$ is also computed correctly.
	\end{lemma}
	\begin{claimproof}
		First, assume that $\DP_x[A,P,F,\rho]$ is \texttt{true}. Then, according to the definition of computation, there exist $P_y$, $Pz$, $F_y$, $F_z$, $\rho_y$, and $\rho_z$ such that $\DP_y[A,P_y,F_y,\rho_y] = \DP_z[A,P_z,F_z,\rho_z] = \mathtt{true}$ and for every $v\in\beta(x)\cap V_I$ we have $\rho_y + \rho_z = \rho + |A|$, $P_y(v) + P_z(v) = P(v)$, and $F(v) = F_y(v) - P_z(v) = F_z(v) - P_y(v)$. By assumption, there exist partial housings $\pi_y$ and $\pi_z$ compatible with $(A,P_y,F_y,\rho_y)$ and $(A,P_z,F_z,\rho_z)$, respectively. We set $\pi_x = \pi_y \cup \pi_z$ and claim that $\pi_x$ is a partial housing compatible with $(A,P,F,\rho)$. Clearly, $\pi_x \subseteq V^x$ as $V^x = V^y \cup V^z$. Moreover, $|\pi_x| = |\pi_y| + |\pi_z| - |A| = \rho_y + \rho_z - |A| = \rho$, as the vertices of $A$ were counted both in $\pi_y$ and $\pi_z$ and the already forgotten vertices that are part of the solution can be part of exactly one partial housing $\pi_y$ or $\pi_z$. By the separator property of the bag $\beta(x)$, we see that in $\pi_x$, all inhabitants assigned to the already forgotten vertices are clearly satisfied with the housing. So, let $v\in\beta(x)$ be a vertex and $h$ be an inhabitant assigned to $v$. First, we verify that the number of past vertices is correct. We have $|(\pi_x\cap N(v))\setminus A| = |(\pi_y\cap N(v))\setminus A| + |(\pi_z\cap N(v))\setminus A| = P_y(v) + P_z(v) = P(v)$, that is, the number of past vertices is correct. Now, we verify that the $v$'s upper-bound is not exceeded. We have $|N(v)\cap \pi_x| + F(v) = P(v) + |A| + F(v) = P_y(v) + P_z(v) + |A| + F_y(v) - P_z(v) = P_y(v) + |A| + F_y(v) \leq \ub(h)$. Therefore,~$\pi_x$ also satisfies the last condition and it follows that $\pi_x$ corresponds to the signature $(A,P,F,\rho)$ in $x$, finishing the left-to-right implication.
		
		In the opposite direction, assume that there is a partial housing $\pi_x$ compatible with some signature $(A,P,F,\rho)$. By the definition of tree decomposition, $V^y \cap V^z = \beta(x)$. We define two partial housings $\pi_y$ and $\pi_z$ such that $\pi_y = \pi_x \cap V^y$ and $\pi_z = \pi_x \cap V^z$. Now, we prove that $\pi_y$ is compatible with some signature $(A,P_y,F_y,\rho_y)$ in $y$ and $\pi_z$ is compatible with some signature $(A,P_z,F_z,\rho_z)$ in $z$, respectively, such that $\rho_y + \rho_z = \rho + |A|$, $P_y(v) + P_z(v) = P(v)$, and $F(v) = F_y(v) - P_z(v) = F_z(v) - P_y(v)$. 
		Let $v\in \pi_x \setminus \beta(x)$. Since $V^y \cap V^z = \beta(x)$, such a past vertex that is in a solution $\pi_x$ is part of exactly one partial housing $\pi_y$ or $\pi_z$. By the definition of $\pi_y$ and $\pi_z$, clearly $\pi_y \subseteq V^y\cap V_U$ and $\pi_z \subseteq V^z\cap V_U$. Moreover, $A \subseteq \beta(x)$ and $\beta(x) = \beta(y) = \beta(z)$ and, therefore, $\pi_y\cap \beta(y) = \pi_z \cap \beta(z) = A$. 
		Now, we focus on the sizes of $\pi_y$ and $\pi_z$. We know that $|\pi_x| = \rho = |A| + |((\pi_x \cap V^x)\setminus A)|$, which can be further decomposed as $|A| + |((\pi_x\cap V^y)\setminus A)| + |((\pi_x\cap V^z)\setminus A)| = |A| + |\pi_y\setminus A| + |\pi_z\setminus A| = |A| + |\pi_y| - |A| + |\pi_z| - |A|$, since we have already shown that $A = \pi_y \cap \pi_z$. Therefore, we have $\rho = |A| + \rho_y - |A| + \rho_z - |A| = \rho_y + \rho_z - |A|$.
		Let $v\in \beta(x)$ be a vertex occupied by an inhabitant $h$. 
		Recall that $\beta(x) = \beta(y) = \beta(z)$. As $\beta(x)$ is a separator in $G^x$, no past vertex is part of both $V^y$ and $V^z$. Hence, clearly $P_y(v) + P_z(v) = P(v)$. Finally, we have $P(v) + F(v) + |A| \leq \ub(h)$. We can see that if we set $F_y(v) = F(v) + P_z(v)$, we obtain $P(v) + F(v) - |A| = P_y(v) + P_z(v) + F_y(z) - P_z(v) - |A| = P_y(v) + F_y(z) - |A| \leq \ub(h)$ and similarly if we set $F_z(v) = F(v) + P_y(v)$. Therefore, the signatures with which $\pi_y$ and $\pi_z$ are consistent with in $y$ and $z$, respectively, satisfy the given conditions. By our assumption, $\DP_y[A,P_y,F_y,\rho_y] = \DP_z[A,P_z,F_z,\rho_z] = \mathtt{true}$ and hence by the definition of computation, $\DP_x[A,P,F,\rho]$ is also set to \texttt{true}. This finishes the proof.
	\end{claimproof}
	
	\bigskip
	To conclude, the correctness of the algorithm follows from \Cref{claim:tw:leaf,claim:tw:introduceEmpty,claim:tw:introduceOccupied,claim:tw:forgetEmpty,claim:tw:forgetOccupied,claim:tw:join} by induction over the nodes of the tree decomposition.
	Moreover, the size of the dynamic programming table for a single node is $2^\Oh{\tw} \cdot n^\Oh{\tw} \cdot n^\Oh{\tw} \cdot \Oh{n} = n^\Oh{\tw}$ and the most time-consuming operation is to compute the value of a single cell in the join node, where it can take $n^\Oh{\tw}$ time, as we need to try all possible functions~$P$. Therefore, we get that the running time of the algorithm is $n^\Oh{\tw}$, which is indeed in \XP. This finishes the proof.
\end{proof}

It is easy to see that the algorithm from \Cref{thm:tw:XP} becomes fixed-parameter tractable if the degree of every vertex of the topology is bounded, as we needed to guess the number of past and future neighbours, which cannot exceed the maximum degree. Hence, we have the following corollary.

\begin{corollary}
	The \ARHshort problem is fixed-parameter tractable when parameterised by the tree-width $\tw$ and the maximum degree $\Delta(G)$ combined.
\end{corollary}

On the other hand, we can show that the algorithm for bounded tree-width graphs cannot be improved to a fixed-parameter tractable one. In fact, we show that the problem is \Wh with respect to this parameter.

\begin{theorem}\label{thm:tw:Wh}
	The \ARHshort problem is \Wh when parameterised by the tree-width $\tw(G)$ of the topology.
\end{theorem}
\begin{proof}
	We reduce from the \textsc{Equitable-3-Colouring} problem. Here, we are given an undirected graph $H$, and our goal is to decide whether there exists a proper colouring of the vertices of $H$ using $3$ colours such that the sizes of two colour classes differ by at most one. This problem is known to be \Wh with respect to the tree-width~\cite{FellowsFLRSST2011}, see also~\cite{MasarikT2020} for further discussion on this result. Without loss of generality, we can assume that $|V(H)|$ is divisible by~$3$.
	
	As the first step in our construction, we create an intermediate graph~$G'$. Initially, the graph $G'$ is the same graph as~$H$. Next, we subdivide each edge~${e\in E(G')}$ and assign to the newly created vertex $w_e$ an inhabitant $h_e$, called \emph{edge-guard}, with $\ub(h_e) = 1$. Next, we create $3$ disjoint copies of the intermediate graph $G'$ which are called $G_1$, $G_2$, and $G_3$, respectively. Let us denote the vertices of a graph $G_i$, $i\in[3]$, corresponding to vertices of the original graph $H$ as $v_1^i,\ldots,v_{|V(H)|}^i$. For every $i\in[3]$, we create a vertex $g_i$ which is connected by an edge to every $v_j^i$, where $j\in[|V(H)|]$, and which is occupied by an inhabitant $h_{g_i}$, called the \emph{size-guard}, with $\ub(h_{g_i}) = \frac{|V(H)|}{3}$. Then, we construct $G$ as a disjoint union of $G_1$, $G_2$, and $G_3$. As the last step of the construction of the graph $G$, we add single vertex $g_{v_i}$ for every vertex $v_i\in V(H)$, make it occupied by an inhabitant $h_{v_i}$, called the \emph{vertex-guard}, with $\ub(h_{v_i}) = 1$ and add the edge $\{g_{v_i},v_i^j\}$ for every $j\in[3]$. Finally, we set $R=|V(H)|$. See \Cref{fig:tw:Wh} for an illustration of the construction.
	
	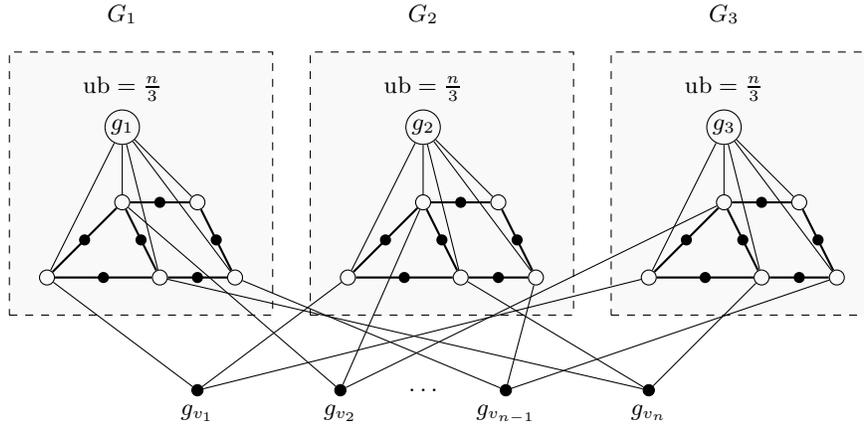
\begin{figure*}[bt!]
		\centering
		\begin{tikzpicture}
			\draw[dashed,fill=gray!5] (-5.5,1) rectangle (-2,-2.5);
			\node at (-4,1.5) {$G_1$};
			
			\node[draw,circle,inner sep=1pt,label={$\ub=\frac{n}{3}$}] (g1) at (-4,0) {$g_1$};
			
			\node[draw,circle,inner sep=2pt] (v11) at (  -5,-2) {};
			\node[draw,circle,inner sep=2pt] (v12) at (  -4,-1) {};
			\node[draw,circle,inner sep=2pt] (v13) at (  -3,-1) {};
			\node[draw,circle,inner sep=2pt] (v14) at (-2.5,-2) {};
			\node[draw,circle,inner sep=2pt] (v15) at (-3.5,-2) {};
			\draw[thick]   (v11) -- (v12) node[midway,draw,circle,fill=black,inner sep=1.2pt] {}
			(v12) -- (v13) node[midway,draw,circle,fill=black,inner sep=1.2pt] {}
			(v13) -- (v14) node[midway,draw,circle,fill=black,inner sep=1.2pt] {}
			(v14) -- (v15) node[midway,draw,circle,fill=black,inner sep=1.2pt] {}
			(v15) -- (v11) node[midway,draw,circle,fill=black,inner sep=1.2pt] {}
			(v12) -- (v15) node[midway,draw,circle,fill=black,inner sep=1.2pt] {};
			
			\draw (g1) edge (v11) edge (v12) edge (v13) edge (v14) edge (v15);

			\draw[dashed,fill=gray!5] (-1.5,1) rectangle (2,-2.5);
			\node at (0,1.5) {$G_2$};
			
			\node[draw,circle,inner sep=1pt,label={$\ub=\frac{n}{3}$}] (g2) at (0,0) {$g_2$};
			
			\node[draw,circle,inner sep=2pt] (v21) at ( -1,-2) {};
			\node[draw,circle,inner sep=2pt] (v22) at (  0,-1) {};
			\node[draw,circle,inner sep=2pt] (v23) at (  1,-1) {};
			\node[draw,circle,inner sep=2pt] (v24) at (1.5,-2) {};
			\node[draw,circle,inner sep=2pt] (v25) at (0.5,-2) {};
			\draw[thick]   (v21) -- (v22) node[midway,draw,circle,fill=black,inner sep=1.2pt] {}
			(v22) -- (v23) node[midway,draw,circle,fill=black,inner sep=1.2pt] {}
			(v23) -- (v24) node[midway,draw,circle,fill=black,inner sep=1.2pt] {}
			(v24) -- (v25) node[midway,draw,circle,fill=black,inner sep=1.2pt] {}
			(v25) -- (v21) node[midway,draw,circle,fill=black,inner sep=1.2pt] {}
			(v22) -- (v25) node[midway,draw,circle,fill=black,inner sep=1.2pt] {};
			
			\draw (g2) edge (v21) edge (v22) edge (v23) edge (v24) edge (v25);
			
			\draw[dashed,fill=gray!5] (2.5,1) rectangle (6,-2.5);
			\node at (4,1.5) {$G_3$};
			\node[draw,circle,inner sep=1pt,label={$\ub=\frac{n}{3}$}] (g3) at (4,0) {$g_3$};
			
			\node[draw,circle,inner sep=2pt] (v31) at (  3,-2) {};
			\node[draw,circle,inner sep=2pt] (v32) at (  4,-1) {};
			\node[draw,circle,inner sep=2pt] (v33) at (  5,-1) {};
			\node[draw,circle,inner sep=2pt] (v34) at (5.5,-2) {};
			\node[draw,circle,inner sep=2pt] (v35) at (4.5,-2) {};
			\draw[thick]   (v31) -- (v32) node[midway,draw,circle,fill=black,inner sep=1.2pt] {}
			(v32) -- (v33) node[midway,draw,circle,fill=black,inner sep=1.2pt] {}
			(v33) -- (v34) node[midway,draw,circle,fill=black,inner sep=1.2pt] {}
			(v34) -- (v35) node[midway,draw,circle,fill=black,inner sep=1.2pt] {}
			(v35) -- (v31) node[midway,draw,circle,fill=black,inner sep=1.2pt] {}
			(v32) -- (v35) node[midway,draw,circle,fill=black,inner sep=1.2pt] {};
			
			\draw (g3) edge (v31) edge (v32) edge (v33) edge (v34) edge (v35);
			
			\foreach[count=\i] \x in {-3,-1.1}{
				\node[draw,circle,fill=black,inner sep=1.5pt,label=270:{$g_{v_\i}$}] (gv\i) at (\x,-3.5) {};
				\draw (gv\i) edge (v1\i) edge (v2\i) edge (v3\i);
			}
			\node at (0,-3.5) {$\hdots$};
			\node[draw,circle,fill=black,inner sep=1.5pt,label=270:{$g_{v_{n-1}}$}] (gv4) at (1.1,-3.5) {};
			\draw (gv4) edge (v14) edge (v24) edge (v34);
			\node[draw,circle,fill=black,inner sep=1.5pt,label=270:{$g_{v_n}$}] (gv5) at (3,-3.5) {};
			\draw (gv5) edge (v15) edge (v25) edge (v35);
		\end{tikzpicture}
		\caption{An illustration of the construction used to prove \Cref{thm:tw:Wh}. All filled vertices are occupied by inhabitants with upper-bound equal to $1$. By $n$ we denote the number of vertices of the original graph $H$.}
		\label{fig:tw:Wh}
	\end{figure*}
	
	Let the input instance of the \textsc{Equitable-3-Colouring} problem be a \YesI and $c\colon V(H)\to[3]$ be an equitable and proper colouring of $H$. Let $V_i = \{v_{i_1},\ldots,v_{i_{\frac{n_H}{3}}}\}$ be the set of vertices with colour $i\in[3]$ according to colouring~$c$. For every $i\in[3]$ and every $j\in[n_H]$, we add to $\pi$ the vertices $v^i_{i_j}$. Now, $\pi$ clearly respects all edge-guards as otherwise $c$ would not be a proper colouring. Also, $\pi$ respects all size-guards, as the size of every colour class is exactly $\frac{n}{3}$, and it is easy to see that all vertex-guards are also satisfied since each vertex is a member of only one colour class.

	In the opposite direction, let $\pi$ be an inhabitants-respecting housing. Recall that $|\pi| = |V(H)|$. The size-guards secure that for every graph $G_i$, $i\in[3]$, we have $|V(G_i)\cap \pi| \leq \frac{|V(H)|}{3}$ and since $|\pi|=|V(H)|$, we get equality.
	
	\begin{claim}\label{lem:good_colouring}
		Let $v_i\in V(H)$ be a vertex of $H$. For every inhabitants-respecting housing~$\pi$ it holds that $|\pi\cap\{v_i^1,v_i^2,v_i^3\}|= 1$.
	\end{claim}
	\begin{claimproof}
		First, suppose that $|\pi\cap\{v_i^1,v_i^2,v_i^3\}| \geq 2$. Then $\pi$ is not inhabitants-respecting since the vertex-guard $h_{v_i}$ approves at most $1$ refugee in the neighbourhood and $h_{v_i}$ is connected with an edge to all $v_i^1$, $v_i^2$, and $v_i^3$. Therefore, $|\pi\cap\{v_i^1,v_i^2,v_i^3\}| \leq 1$. If $|\pi\cap\{v_i^1,v_i^2,v_i^3\} = 0$, then, by a simple counting argument, there is a vertex $v_j\in V(H)$, $v_j \not= v_i$, such that $|\pi\cap\{v_j^1,v_j^2,v_j^3\}| \geq 2$, which is not possible by previous argumentation.
	\end{claimproof}
	
	Now, we define a solution colouring for $G$. For every $v_i\in V(H)$, we set
	\[
	c(v_i) = \begin{cases}
		1 & \text{if $v_i^1\in\pi$,}\\
		2 & \text{if $v_i^2\in\pi$,}\\
		3 & \text{if $v_i^3\in\pi$.}
	\end{cases}
	\]
	Note that thanks to \Cref{lem:good_colouring} the colouring $c$ is well-defined and we already argued that $c$ is equitable. What remains to show is that $c$ is a proper colouring. For the sake of contradiction, assume that there are $v_i$ and $v_j$ of the same colour and, at the same time, $e = \{v_i,v_j\}\in E(H)$. Without loss of generality, let $c(v_i) = c(v_j) = 1$. In the graph $G_1$, there is the vertex $w_{e}$ occupied by an inhabitant $h_{e}$ with $\ub(h_e) = 1$ which is neighbour of both $v_i^1$ and $v_j^1$. But $v_i$ and $v_j$ can be coloured with the same colour if and only if both $v_i^1$ and $v_j^1$ are in $\pi$, which is not possible as $\pi$ would not be inhabitants-respecting. Thus, the colouring $c$ is a proper colouring of $H$.
	
	Finally, we show that tree-width of our construction is indeed bounded. Let $\mathcal{T}=(T,\beta,r)$ be a tree decomposition of $H$. We construct the tree decomposition $\mathcal{T}'$ of the graph $G$ as follows. The tree $T$ and the root $r$ remain the same. Now, let $x\in V(T)$ be a node of $T$. We define a new mapping $\beta'\colon V(T)\to 2^{V(G)}$ such that we replace all $v_i\in\beta(x)$ with vertices $v_i^1$, $v_i^2$, $v_i^3$, and $g_{v_i}$. As the final step, we add to each bag the vertices $g_1$, $g_2$, and $g_3$. By construction, the width of the tree decomposition $\mathcal{T}'$ is at most $4\cdot\tw(H) + 3$, which finishes the proof.
\end{proof}

For problems that are not tractable when restricted to bounded tree-width graphs, it is natural to seek tractability with respect to more restricted parameters. A natural candidate for such a parameter is the vertex-cover number~\cite{KorhonenP2015,GanianHKRSS2023,BlazejGKPSS2023}; however, fixed-parameter tractability of \ARHshort with respect to the vertex-cover number follows, again, from the more general setting~\cite{KnopS2023}.

Therefore, we turn our attention to topologies of a bounded feedback-edge number. Recall that the feedback-edge number is the minimum number of edges that need to be removed from a graph to obtain an acyclic graph. Such networks appear naturally in practice, for example, in backbone infrastructure, waterway monitoring, or transportation maps~\cite{BentertBNN2017}. Also, the feedback-edge set number has been used successfully to obtain fixed-parameter algorithms in many important problems for AI research~\cite{GanianK2021,DementievFI2022,GruttemeierK2022,BredereckHKN2022}.

\begin{theorem}
	The \ARHshort problem is fixed-parameter tractable when parameterised by the feedback-edge set number $\fes(G)$ of the topology.
\end{theorem}
\begin{proof}
	Let $k=\fes(G)$. Since the topology $G$ is a bipartite graph with one part consisting of refugees and the second part consisting of empty houses, every edge~$f$ in a feedback-edge set $F$ connects an inhabitant and an empty vertex. Therefore, $F$ contains $\frac{k}{2}$ empty vertices. Let $F_U$ be the set of empty vertices in $F$. Our algorithm guesses a set $S\subseteq F_U$ of empty vertices for a solution. If $|S| > R$, we refuse the guess and continue with another possibility. Otherwise, we create the reduced instance $\mathcal{I}'=((V(G)\setminus F_U, E(G)\setminus F),\I,\assgn,\R-|S|,\ub')$, where for every $j\in I$ we set $\ub'(h) = \ub(h) - |S\cap N_G(\assgn(h))|$. If this produces an inhabitant with negative upper-bound, we again refuse the guessed $S$. Otherwise, we apply \Cref{lem:remove-intolerant}. Finally, since the topology of $\mathcal{I}'$ is by definition a forest, we try to solve $\mathcal{I}'$ using \Cref{thm:forest:poly}. If there exists a solution for $\mathcal{I}'$, then there is also a solution for the original $\mathcal{I}$.
	
	There are $2^{\frac{k}{2}}$ possible subsets $S$ and every guess can be solved in polynomial time using \Cref{thm:forest:poly}. This gives us the total running time $2^\Oh{k}\cdot n^\Oh{1}$, which is indeed in \FPT, and the theorem follows.
\end{proof}

\subsubsection{Dense Topologies.}

\noindent{}So far, we have studied mostly topologies that are, from a graph-theoretical perspective, sparse. In what follows, we turn our attention to dense graphs. Recall that all topologies are naturally bipartite and, therefore, the palette of structural parameters that are bounded for dense graphs is a bit restricted.

We start this line of research with a simple algorithm for complete bipartite graphs. In fact, for complete bipartite graphs, we only need to compare the minimum upper-bound with the number of refugees.

\begin{proposition}
	Every instance of the \ARHshort problem where the topology is a complete bipartite graph can be solved in linear time.
\end{proposition}
\begin{proof}
	Our algorithm is based on a simple comparison. By enumerating all inhabitants, we can compute $\ub_{\min}$, which is the minimum upper-bound over all inhabitants. Now, if $\R \leq \ub_{\min}$, then we return \Yes. Otherwise, the result is \No as there is at least one inhabitant who does not approve $\R$ refugees in his or her neighbourhood.
\end{proof}

As the situation with complete bipartite graphs was rather trivial, we turn our attention to graphs that are not far from being complete bipartite. In particular, we will study graphs that can be obtained from complete bipartite graphs by deleting~$p$ edges. Formally, this family is defined as follows.

\begin{definition}
	A graph $G$ is \emph{$p$-nearly complete bipartite} if it can be obtained from a complete bipartite graph by the removal of $p$ edges.
\end{definition}

\begin{theorem}
	The \ARHshort problem where the topology is $p$-nearly complete bipartite graph is fixed-parameter tractable when parameterised by~$p$.
\end{theorem}
\begin{proof}
	We denote by $\I_C$ the set of inhabitants that are neighbours of all empty houses, that is, $\I_C = \{h\in \I \mid N(\assgn(i)) = V_U\}$. If $|\I_C|$ is at least one and $\min_{h\in\I_C} \ub(h) \leq \R$, then we return \No.
	
	Now, let $p < |\I|$. Then there is at least one inhabitant in $\I_C$. As we have $\min_{h\in\I_C} \ub(h) > \R$, inhabitants in $\I_C$ approve any housing. Therefore, we can safely remove them from $G$ without changing the solution. What remains is an instance with at most~$p$ inhabitants which can be resolved in \FPT time using \Cref{lem:ARH:FPT:I}.
	
	If it holds that $p \geq |\I|$ we have that the number of inhabitants is small ($\Oh{p}$, in fact). Therefore, we can again use \Cref{lem:ARH:FPT:I} to decide the instance directly in \FPT time.
\end{proof}

Using similar but more involved ideas, we can also show that the problem of our interest is tractable if the topology is not far from being a complete bipartite graph in terms of the number of removed vertices.

\begin{theorem}
	The \ARHshort problem is fixed-parameter tractable when parameterised by the distance to complete bipartite graph.%
\end{theorem}
\begin{proof}
	Let $M$ be a modulator to a complete bipartite graph of size~$k$ and let $G\setminus M$ be a complete bipartite graph. We first guess a set of vertices $S\subseteq M\cap V_U$ that are housing refugees in a solution. Next, we decrease the upper-bound of each inhabitant assigned to the neighbourhood of $S$ and remove $M\cap V_U$ from $G$. Then we replace all vertices in $V_I \setminus M$ with a single vertex $g$ occupied by an inhabitant $h$ with an upper-bound equal to the minimum over all inhabitants occupying vertices outside the modulator; formally, $\ub(h) = \min_{v\in V_I\setminus M} \ub(\assgn^-1(v))$. It is easy to see that in the reduced instance we have $\Oh{k}$ inhabitants. Therefore, we can use \Cref{lem:ARH:FPT:I} to decide whether it is possible to house $R-|S|$ refugees in houses of the reduced instance.
	
	There are $2^\Oh{k}$ possible sets $S$. The reduction of the instance is polynomial-time and deciding the reduced instance is in \FPT, which finishes the proof.
\end{proof}

\section{Relaxing the Upper-Bounds}\label{sec:relaxed}
\newcommand{\exc}{\ensuremath{\operatorname{exc}}}
In previous section, we discussed the model where we were interested in housings that perfectly meet upper-bounds of all inhabitants. Despite the fact that such outcomes are perfectly stable, this requirement, as illustrated in \Cref{lem:housing-non-existence,lem:housing-non-existence-intolerant} can lead to negative outcomes even for fairly simple settings.

In this section, we relax the notation of inhabitants-respecting housing by allowing small violation of inhabitants' upper-bounds. Before we define the corresponding computational problem formally, we introduce some auxiliary notation. Let~$h\in\I$ be an inhabitant,~$\pi$ be a set of vertices housing refugees, and~$v\in V$ be a vertex such that~$\assgn(h) = v$. The function~$\exc_\pi\colon\I\to\N$ is a function that assigns to each inhabitant the number of refugees housed in his neighbourhood exceeding~$\ub(h)$, that is,~$\exc_\pi(h) = \max\{0,|N_G(\assgn(h))\cap\pi|-\ub(h)\}$. If the housing $\pi$ is clear from the context, we drop the subscript. Slightly abusing the notation, we set $\exc(\pi) = \sum_{h\in\I} \exc_\pi(h)$, where $\pi$ is a housing. We say that housing $\pi\subseteq V_U$ is \emph{$t$-relaxed inhabitants-respecting} if $\exc(\pi) \leq t$.

\vspace{0.15cm}\noindent
\begin{tabularx}{\linewidth}{lX}
	\toprule
	\multicolumn{2}{c}{\small\textsc{Relaxed} \ARHshort (\RARHshort)} \\\midrule
	\small\emph{Input:} & \small{}A topology $G=(V,E)$, a set of inhabitants $\I$, a number of refugees~$\R$, an inhabitants assignment $\assgn\colon\I\to V$, an upper-bound $\ub\colon\I\to\N$, and a number $t$. \\
	\small\emph{Question:} & \small{}Is there a $t$-relaxed inhabitants-respecting housing~$\pi$ of size $\R$? \\
	\bottomrule
\end{tabularx}
\vspace{0.15cm}

To illustrate the definition of the problem, what follows is a simple example of an \RARHshort instance with $t=2$.

\begin{example}\label{ex:model2}\itshape
	Let $\mathcal{I}$ be as shown in \Cref{fig:example2}, that is, we have three inhabitants $h_1$, $h_2$, and $h_3$ with the same upper-bound $\ub(h_1) = \ub(h_2) = \ub(h_3) = 1$, and we have three refugees to house. The empty vertices are $x$, $y$, $z$ and $w$. Additionally, we set $t=2$.
	\begin{figure}[tb!]
		\centering
		\begin{tikzpicture}[scale=0.75]
			\node[draw,circle,label={270:\small$\ub(h_1) = 1$}] (h1) at (0,0) {$h_1$};
			\node[draw,circle,label={90:\small$\ub(h_2) = 1$}] (h2) at (4,1.75) {$h_2$};
			\node[draw,circle,label={270:\small$\ub(h_3) = 1$}] (h3) at (8,0) {$h_3$};
			\node[draw,circle,label={270:$y$},fill=black] (u1) at (4,0) {};
			\node[draw,circle,label={90:$w$},fill=black] (u4) at (8,1.75) {};
			
			\draw (h1) -- (u1) -- (h3);
			\draw (h2) -- (u1);
			\draw (h2) -- (u4);
			\draw (h1) -- (h2) node[midway,draw,circle,fill=black,label={90:$x$}] {};
			\draw (h2) -- (h3) node[midway,draw,circle,fill=black,label={90:$z$}] {};
		\end{tikzpicture}
		\caption{An example of the \RARH problem with $t=2$. The notation is the same as in \Cref{fig:example}.}
		\label{fig:example2}
	\end{figure}
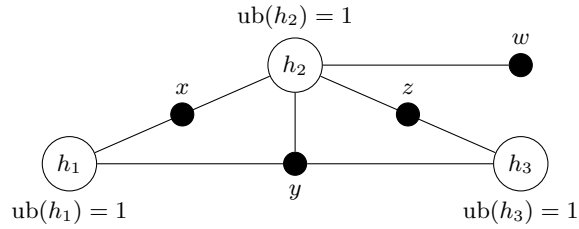
	Clearly, there is no inhabitants-respecting housing as the upper-bound of inhabitant $h_2$ is equal to $1$, and he is the neighbour of all empty vertices. Let $\pi_1 = \{x,y,z\}$ be a housing. Then $\exc(h_1) = \exc(h_3) = 1$ as both of them are neighbouring two vertices of $\pi_i$ and their upper-bound is exactly one. For $h_3$, we have $\exc(h_3) = 2$ Thus $\exc(\pi_1) = 4$.
	On the other hand, if we choose as a solution $\pi_2$ the vertices $x$, $z$, and~$w$, we obtain that for inhabitants $h_1$ and $h_3$ the number of refugees in their neighbourhood matches their upper-bounds. Clearly, this is not the case for $h_2$ as he is a neighbour of all refugees and his upper-bound is also $1$, that is, $\exc(h_2) = 2$. Overall, we have $\exc(\pi_2) = \exc_{\pi_2}(h_1) + \exc_{\pi_2}(h_2) + \exc_{\pi_2}(h_3) = 0 + 2 + 0 = 2$.
\end{example}

Now, we turn our attention to the existence guarantees we have in this relaxed setting. Observe that for~$t=0$, the~$t$-\RARHshort problem is, in fact, equivalent to the standard \ARHshort problem studied in the previous sections. Therefore, the following result is not very surprising.

\begin{theorem}
	\label{thm:tARH:NPc}
	It is \NPc to decide the \RARH problem for any~$t\geq 0$.
\end{theorem}
\begin{proof}
	The reduction is similar to the one used in the proof of \Cref{thm:NPh}. The only difference is that as our first step, we modify the graph~$G$ from the instance of the \textsc{Independent Set} problem such that we replace every edge~$e_i$,~$i\in[m]$, with~$t+1$ parallel edges~$e_i^1,\ldots,e_i^{t+1}$, and we add~$t$ additional paths on two vertices with one vertex occupied by an intolerant inhabitant. Finally, we increase the number of refugees to $\R = k + t$.
	
	The first modification ensures that the violation of the upper-bound cannot occur on the vertices of the graph~$G_1$, and~$t$ additional paths provide space for problematic refugees. The correctness and the running time remain completely the same.
\end{proof}

It should be noted that the construction used in the proof of the previous theorem does not produce a topology of constant degree; the maximum degree depends on the value of~$t$. However, it clearly shows that for small values of $t$, we cannot expect any stronger tractability result. However, some of the known algorithms can be modified to work for \RARHshort, such as the following.

\begin{theorem}
	The \RARH problem is in \XP when parameterised by the tree-width \tw{} of the topology.
\end{theorem}
\begin{proof}[sketch]
	The algorithm is very similar to the one given in \Cref{thm:tw:XP}. We add only two more indices to the dynamic programming table. The first index, say $\tau$, is a function that represents how much the upper-bounds of the bag vertices will be exceeded in a hypothetical solution. The second index is a single number that represents how much we relaxed the upper-bounds of past vertices.
\end{proof}

The natural question is what happens with large values of~$t$. In the following proposition, we show that there exists a bound that guarantees the existence of a $t$-relaxed inhabitants-respecting housing. 

\begin{proposition}
	\label{lem:tARH:bound}
	Every instance of $t$-\RARHshort, where $t \geq |\I|\cdot\R$, admits a $t$-relaxed inhabitants-respecting refugee housing.
\end{proposition}
\begin{proof}
	Let $\pi$ be an arbitrary housing of refugees on empty vertices of the topology. Every inhabitant $h$ can have at most $\R$ refugees housed in her neighbourhood and therefore $\exc(h) \leq R$. As there are $|\I|$ inhabitants in total, setting $t = |\I|\cdot\R$ always leads to a $t$-relaxed inhabitants-respecting housing.
\end{proof}

Surprisingly, the trivial bound given in \Cref{lem:tARH:bound} is tight in the sense that there are instances where $t=|\I|\cdot\R$ is necessary for the existence. Let the topology be a complete bipartite graph $G=(A\cup B,E)$, where $|A|=|\I|$ and $|B|=\R$, all inhabitants are assigned to $A$, and for every $h\in\I$ let $\ub(h) = 0$. As we are allowed to assign refugees only to $B$ and each inhabitant is neighbour of every vertex $v\in B$, there is no $t$-relaxed inhabitants-respecting refugee housing for any $t < |\I|\cdot\R$.

The question that arises is whether the complexity of the problem changes for large, but still smaller than the trivial guarantee given in \Cref{lem:tARH:bound}, values of $t$. We answer this question positively in our last result.

\begin{theorem}
	\label{thm:tARH:belowBound}
	There is an \XP algorithm with respect to~$q$ deciding whether a given \RARHshort instance admits a $(|\I|\cdot\R - q)$-relaxed inhabitants-respecting refugee housing.
\end{theorem}
\begin{proof}
	As the first step of our algorithm, we guess a $q'\leq q$-sized set $S=\{s_1,\ldots,s_{q'}\}$ of inhabitants such that for each $s_i$, $i\in[p']$, we have $\exc(s_i) < \R$. Next, for every $s_i\in S$, we guess a number $\operatorname{eff}(s) \in [\R-q,\R-1]$ that represents the effective value of $\exc$ in the solution. Finally, we remove all inhabitants $h\in \I\setminus S$ together with their vertices, set $\ub'(s) = \ub(s) + \operatorname{eff}(s)$ for every $s\in S$, and solve the reduced instance, which now has at most $q$ inhabitants, using the \FPT algorithm from \Cref{lem:ARH:FPT:I}. The total running time is $n^\Oh{q} \cdot q^\Oh{q}$ for guessing and \FPT for solving, which is indeed in \XP.
\end{proof}

To the best of our knowledge, \Cref{thm:tARH:belowBound} is one of the very few below-guarantee parameterisations in AI and computational social choice -- despite the fact that computational social choice problems lie at the very foundation of this research direction~\cite{FernauFLMPS2014,MahajanRS2009}.

\section{Conclusions}

We initiate the study of restricted preferences in the context of the refugee housing problem. Our main motivation is that the preferences studied in the original setting are not well-aligned with real-life scenarios, and, moreover, these properties were exploited to give intractability results in~\cite{KnopS2023}. In our model, we not only introduce tractable algorithms, but also unveil intractability results, shedding light on the complexity landscape of the problem. 

For a complete understanding of the computational picture of \ARHshort, we still miss some pieces. In particular, in \Cref{thm:remaining:XP,thm:tARH:belowBound} we give only \XP algorithms. The natural question then is whether there exist \FPT algorithms for these settings. Also, our selection of structural restrictions is done mostly from the theorists' perspective; they are either widely studied (tree-width) or naturally emerge from the model's properties (dense parameters). The only exception is the feedback-edge set number. Therefore, we see great potential in the analysis of real-life topologies and the investigation of their structural properties.

We also recognise the need to explore more intricate models of refugee housing that capture the nuances of strategic behaviour and agent interactions. Inspired by Schelling games~\cite{AgarwalEGISV2021}, we see the potential to uncover novel dynamics by allowing agents to swap positions or jump to empty houses. This direction holds promise for deepening our understanding of real-life scenarios and the complex decisions faced by displaced individuals.

\bibliographystyle{splncs04}
\bibliography{references}

\end{document}